\documentclass[12pt]{article}    
\usepackage[linesnumbered,ruled]{algorithm2e}

\usepackage{natbib}
\usepackage[english]{babel}
\usepackage{amsthm}
\usepackage{subfigure}
\usepackage{amsfonts}
\usepackage{amsmath}
\usepackage{color}
\usepackage{tikz,comment}
\usetikzlibrary{shapes,snakes}
\usetikzlibrary{backgrounds,calc}
\usepackage{enumitem}
\usepackage{pgfplots}
\usepackage[affil-it]{authblk}

\usetikzlibrary{external}
\newtheorem{thm}{Theorem}
\newtheorem{lem}{Lemma}

\newtheorem{prop}{Proposition}
\newtheorem{dfn}{Definition}
\newtheorem{ex}{Example}

\pdfoutput=1

\title{Identification of Essential Proteins Using Induced Stars in Protein-Protein Interaction Networks}
\author{Chrysafis Vogiatzis%
  \thanks{E-mail: \texttt{cvogiatzis@ncat.edu}; Corresponding author}}
\affil{Department of Industrial \& Systems Engineering,\\ North Carolina A\&T State University,\\ Greensboro, NC, USA}

\author{Mustafa Can Camur%
  \thanks{E-mail: \texttt{camurm@rpi.edu}}}
\affil{Department of Industrial \& Systems Engineering,\\ Rensselaer Polytechnic Institute,\\ Troy, NY, USA}
\date{}

\begin{document}
\maketitle
\begin{abstract}%
In this work, we propose a novel centrality metric, referred to as \emph{star centrality}, which incorporates information from the closed neighborhood of a node, rather than solely from the node itself, when calculating its topological importance. More specifically, we focus on \emph{degree} centrality and show that in the complex protein-protein interaction networks it is a naive metric that can lead to misclassifying protein importance. For our extension of degree centrality when considering stars, we derive its computational complexity, provide a mathematical formulation, and propose two approximation algorithms that are shown to be efficient in practice. We portray the success of this new metric in protein-protein interaction networks when predicting protein essentiality in several organisms, including the well-studied \emph{Saccharomyces cerevisiae}, \emph{Helicobacter pylori}, and \emph{Caenorhabditis elegans}, where star centrality is shown to significantly outperform other nodal centrality metrics at detecting essential proteins. We also analyze the average and worst case performance of the two approximation algorithms in practice, and show that they are viable options for computing star centrality in very large-scale protein-protein interaction networks, such as the human proteome, where exact methodologies are bound to be time and memory intensive.
\smallskip


\end{abstract}%

\section{Introduction}

Protein-protein interaction networks are mathematical constructs where every protein is represented by a vertex, with two vertices connected by an edge whenever the corresponding proteins interact. Typically, with every edge we associate a {\em weight}, that captures the strength of the interaction. These constructs have enabled complex network analysis and graph theoretic tools in purely biological problems. For example, we now possess novel computational tools to detect protein complexes \citep{li2010computational,mitra2013integrative}, predict protein roles and essentiality \citep{typas2015bacterial,ren2011prediction,li2010network}, among others. 

Another advancement that has led to an increasing interest in such biological problems is the availability of large-scale biological data. Nowadays, there are multiple databases containing information based on years of biological experimentation on protein interactions. Indicatively, we mention the curated collections of proteomic data made readily available by \cite{franceschini2013string,szklarczyk2014string}, \cite{pagel2005mips}, and \cite{salwinski2004database}. In essence, complex network theory and tools, coupled with an unprecedented growth in the data available for analysis has led to significant scientific interest in this field of computational biology focused on the study of protein-protein interaction networks.


The challenge we aim to tackle in this work can be summarized as follows: {\em does there exist a network topology metric that captures the importance of a single protein in the grand scheme of the proteome}? We use the term {\em proteome} to describe the complete universe of proteins and their interactions in an organism. The challenge we are focusing on is not new, as it has attracted numerous researchers and has led to the investigation of various metrics, ranging from graph modularity \citep{narayanan2011modularity} to centrality \citep{hahn2005comparative}. Being able to use such objective metrics for studying the proteome is of importance, as it can lead us to the detection of informal groups in the interaction network \citep{pereira2004detection}. 

With the term ``{\em detection of informal groups}" we mean the detection of sets or clusters of proteins, based only on their interactions and the topological structure, and no other externally available information. Such detection techniques would enable us with objective methods of measuring protein importance in the proteome independently of other biological experiments and could guide future experimentation. In general, topological importance (also broadly referred to as \emph{centrality}) is a well-studied topic in complex networks, including protein-protein interaction networks. In our work, though, we propose a novel centrality metric for each protein in the network. This metric aims to capture both the individual interactions of every protein, as well as the interactions of its open neighborhood, when disregarding neighboring nodes that are connected to one another, hence forming an induced star. We refer to this centrality as \emph{star centrality}. 

\subsection{Outline}

We first provide a review of protein essentiality, along with the definition of ``party" and ``date" hubs. In the same part, we also discuss previous computational tools, both centrality-based and others, in detecting essential proteins. In Section \ref{notation}, we present the basic notation we will be using throughout the paper, define the problem, and provide its computational complexity. Then, Section \ref{model} focuses on our mathematical programming framework; in the same section, we propose greedy heuristic approaches for tackling the problem faster and provide their approximation guarantees. Section \ref{results} (supplemented by the Appendix of the paper) presents our computational study on five protein-protein interaction networks, namely \emph{Saccharomyces cerevisiae} (yeast), \emph{Helicobacter pylori}, \emph{Staphylococcus aureus}, \emph{Salmonella enterica CT18}, and \emph{Caenorhabditis elegans}. The performance of the approximation algorithms is also contrasted to the exact solution. We conclude this work with our observations and our insights in Section \ref{conclusions}. 

\subsection{Protein-protein interaction networks}

Protein-protein interaction networks (PPINs) have become, mostly over the last decade, an important point of discussion in our quest to better understand and analyze how and why proteins interact with one another. As proteins are fundamental entities that control numerous biological activities, information on how they bind and interact to perform said activities is an important scientific endeavor that can bring to light insight into cell mechanisms.

Before proceeding to the main body of the related literature, we briefly discuss how and where PPINs are made available. The first step towards creating a PPIN is to experimentally discover and validate pairs of proteins that interact. Even though there exists a wide range of genetic and biochemical tools to detect such interactions \citep{peng2016protein}, two of the most common systems are yeast two hybrid (Y2H) \citep{sardiu2011building} and coaffinity purification and mass spectrometry (AP/MS) systems \citep{teng2014network}. After further analysis, a collection of the identified interactions comprises the overall network that can be used. PPINs are now readily available from many different databases, such as the ones by \cite{xenarios2000dip,zanzoni2002mint,pagel2005mips,franceschini2013string,chatr2013biogrid}, among others. However, it has been observed that such networks are unfortunately not without errors \citep{legrain2000genome,sprinzak2003reliable,hart2006complete}.

A fundamental question in the analysis of PPINs (as well as in general biological and other networks) is whether there exist proteins (entities) that can significantly alter cell functionality (or, can even cause cell death). A protein is said to be \emph{essential} or \emph{lethal} when, if absent, it causes the biological cell to die \citep{kamath2003systematic} or prevents it from properly reproducing. Such proteins are indispensable for growth and development and identifying them is important for better understanding the minimal requirements for cell life \citep{Acencio2009}. Moreover, essential proteins provide insight in human gene morbidity: \cite{wilson1977biochemical} showed that proteins encoded by essential genes evolve slower than their non-essential counterparts, while \cite{doi:10.1093/nar/gkh330} proved the existence of similarities between human morbid genes and essentiality in Drosophila melanogaster (fruit fly). The study of essential proteins was and still is typically performed experimentally; however, those experiments tend to be expensive, both resource- and time-wise \citep{tang2014predicting}. Examples of such experimental techniques include conditional gene knockout \citep{skarnes2011conditional} and RNA interference \citep{cullen2005genome}. Nowadays, with the availability of vast amounts of proteomic data, information on essentiality of proteins is also increasing: for instance, we refer the reader to the curated Database of Essential Genes, or DEG \citep{zhang2004deg,zhang2009deg,luo2013deg}.

It has been observed that the study of protein essentiality can be targeted to only a select number of proteins (or, equivalently, proteins can be discarded from contention) using quantitative techniques. In 2001, \cite{jeong2001lethality} introduced the {\em centrality-lethality} rule, where lethality can be used as proxy for essentiality. This was the first work to make the observation that network topology (and specifically, centrality) can suggest essentiality; therein, the authors focus on degree centrality and show that it is a good indicator of protein essentiality. \cite{hahn2005comparative} were able to show that there is a relationship between the position of a protein in the network and its evolution rate. Seeing as essential genes do tend to evolve slower \citep{wilson1977biochemical} and essential proteins are products of essential genes, proteins that are centralized, regardless of degree, are more prone to being essential. The interested reader is also referred to the work by \cite{zotenko2008hubs} which introduces a new explanation on why hubs tend to be essential: their essentiality stems from their participation in {\em groups} of highly connected proteins that are also enriched in essential proteins. A \emph{hub} is defined as a protein with many interactions. Seeing as this definition is very open-ended, some researchers use different threshold values for the number of interactions. For instance, in the work of \cite{han2004evidence}, a hub is defined as a protein with more than 5 interactions. The question of why hubs appear more prone to being essential is also investigated by \cite{he2006hubs}.

In general, {\em degree centrality} (or the number of interactions of a protein) has been investigated in a series of works \citep{jeong2001lethality,han2004evidence,yu2004genomic}. Other centrality metrics that have been investigated over the years include {\em betweenness centrality} \citep{joy2005high}, {\em closeness centrality} \citep{wuchty2003centers}, {\em bipartivity} \citep{estrada2006protein}. In a comprehensive computational study by \cite{estrada2006virtual}, it was shown that selecting the top ranked proteins according to different centrality metrics always results in a better predictor of essentiality than the random selection. 

On top of the above centrality-based approaches, a number of techniques have recently been proposed that aim to incorporate different metrics and network characteristics. As an example, ``bottleneck'' proteins and their relationship to essential proteins is investigated by \cite{yu2007importance}. In their contribution, \cite{ren2011prediction} predict essential proteins by incorporating information from the subgraphs and protein complexes each protein belongs to. In the spirit of weighing information from different sources and metrics, \cite{chua2008unified} propose an integrated, unified weighing scheme to detect protein essentiality. Other weighing schemes that balance information from both centrality and other network topology metrics are due to \cite{li2010essential}, \cite{li2014effective}, \cite{jiang2015essential}. Last, balancing information from a variety of sources in order to cancel out the effects of erroneous data present is proposed by \cite{tang2014predicting}.



\cite{han2004evidence} also investigate another, very important protein characterization, one between proteins that interact with all their neighbors simultaneously and ones that interact with their partners in different times and/or locations, also referred to informally as ``party" and ``date" hubs, respectively. More formally stated, ``party" hubs show high co-expression with their partners, while ``date" hubs exhibit the opposite. An example of how the definition of ``party" and ``date" hubs would look like in a toy network is shown in Figure \ref{partyVSdate}. This computational discovery has been met with scrutiny by the scientific community and has led to a general debate on whether this classification of proteins actually helps us decode the proteome \citep{mirzarezaee2010features}. In general, though, this hypothesis has led to significant interest in connecting graph theoretic notions to PPIN analysis (see, e.g., the works by \cite{agarwal2010revisiting} and \cite{gursoy2008topological}, among others). 

\begin{figure}[htbp]
\centering
{\includegraphics{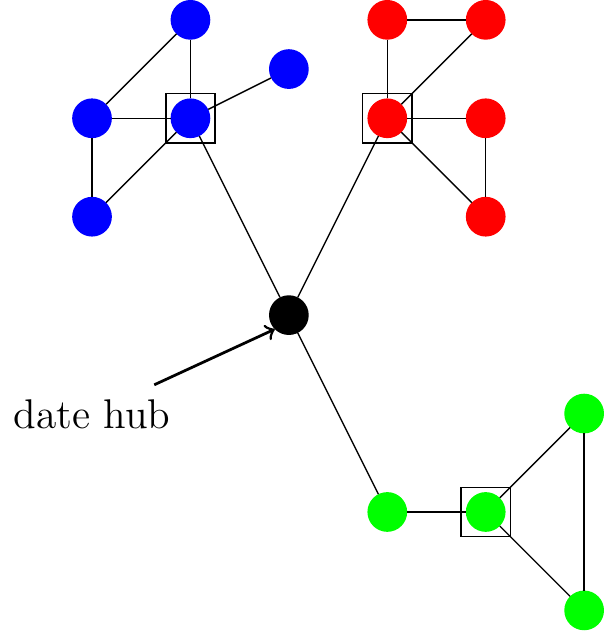}}
\caption{An example of how ``party" and ``date" hubs would appear as in a PPIN. The nodes of blue, red, and green color represent three different structures/complexes in a generated PPIN. The ``party" hubs are marked with a square, while the ``date" hub is annotated on the Figure. Observe that the ``date" hub possesses a smaller number of interactions (smaller degree) than some of the ``party" hubs in this example. \label{partyVSdate}}
\end{figure}


A specific extension that is of interest to us has to do with \emph{group} centrality. Recently, we have seen more work that focuses on extending centrality notions to a group of nodes in the network \citep{everett1999centrality,everett2005extending,borgatti2006identifying}. This extension enables us with notions of endogenous and exogenous centrality \citep{everett2010induced}, where a network property is taken and measured after node/edge deletion, and also provides us with a tool to consider clusters of nodes and figure out their topological importance. An integer programming formulation for detecting informal, cohesive groups with high and low centrality was presented by \cite{vogiatzis2015integer}. 

Seeing as centrality has been a recurring theme in the study of biological networks, and more specifically, PPINs, we propose to investigate {\em group centrality} in this context. Centrality has indeed proven an important characteristic of PPINs, despite the existing caveats with nodal metrics. First, assigning importance to a single protein (resp. interaction), instead of a set of proteins (resp. interactions) tends to favor those proteins that participate in large, dense complexes. Secondly, the datasets of PPINs are still not error-free \citep{hart2006complete}; assuming complete information can lead to significant misattributions of importance. Last, some proteins that present low co-expression with their interacting partners would be disregarded by such metrics even though they might have a significant role in coordinating different complexes (e.g., ``date" hubs). We will then contrast the performance of our proposed metric to nodal centrality metrics (degree, betweenness, closeness, eigenvector), while at the same time, showing that it alleviates all the above issues. We can now proceed to formally state the notation and the definition of the problem in the next section.

\section{Fundamentals}
\label{notation}

Let $G(V, E)$ represent a simple, undirected graph with a vertex set $V$ of size $|V|=n$ nodes and an edge set $E\subset V\times V$ of size $|E|=m$. We say that two nodes $i,j\in V$ are connected by an edge if the adjacency matrix entry $a_{ij}=1$; otherwise we have that $a_{ij}=0$. Seeing as the graphs considered here are undirected, the adjacency matrix is symmetric. We further consider a positive weight parameter on the edges of the graph, $w_e: E\mapsto \mathbb{R}, \forall e=(i,j)\in E$. Furthermore, the open neighborhood of a node $i\in V$ is defined as $N(i)=\{j\in V: (i,j)\in E\}$; similarly, the closed neighborhood of a node $i$ is defined as $N[i]=N(i)\cup \{i\}$. The definition can be extended to apply for sets of nodes $S\subseteq V$, as $N(S)=\{j\in V\setminus S: (i,j)\in E \textit{ for some } i\in S\}$. The notion of (open) neighborhood is sometimes generalized to include nodes that are reachable within at most $k$ hops. This neighborhood is represented here by $N^k(i)$: for example, the complete set of nodes reachable by $i\in V$ within at most 2 hops would be denoted as $N^2(i)$. Using the above definitions, node degree centrality can be easily represented as $$\mathcal{C}^d(i)=|N(i)|.$$

We also define the subgraph induced by a set of nodes $S$, $G[S]$ as the subgraph of $G$ with a vertex set $V[G[S]]=S$ and an edge set $E[G[S]]=\{(i,j)\in E: i, j\in S\}$. We further say that a set of nodes $S$ forms an \emph{induced star} if the induced subgraph of $S$ has exactly one node of degree $|S|-1$ and $|S|-1$ nodes of degree 1. 

\subsection{Problem definition}

In this work, we define a centrality measure that incorporates information from the centrality of the open neighborhood, instead of relying solely on the considered node. More specifically, we focus on degree centrality:

\begin{dfn}
The star degree centrality of a node $i$ is the degree centrality of the induced star $S$ centered at $i$ that produces the maximum open neighborhood size of $S$.
\end{dfn}

Formally, this can be expressed as in \eqref{starcentrality}.

\begin{align}
	\label{starcentrality}
	\mathcal{C}^s (i) = \max\{|N(S)|: S\subseteq V \textit{ forms an induced star centered at } i\in V\}
\end{align}

We proceed to provide two examples for better problem representation. They are described in Examples \ref{partyVSdateExample} and \ref{partyVSdateExample2}.

\begin{ex}
\label{partyVSdateExample}
As an example, let us return to the graph of Figure \ref{partyVSdate}. Consider first the portrayed date hub in the middle. There are eight possible induced star configurations with the date hub at their center (including the configuration consisting of the center itself with no other nodes as ``leaves''). As an example, consider the configuration consisting of the center itself, the neighboring blue node, and the neighboring green node. The open neighborhood size of this star configuration is equal to 6 (the four blue neighbors, the red neighboring node, and the green neighbor). Doing this for all possible configurations reveals that the induced star centered at the date hub, that produces the maximum open neighborhood size would be either the set $S$ consisting of the date hub, and the blue and red party hubs, or it could also include the green node in the lower right connection of the date hub (both have an open neighborhood size of 9). Hence, the star centrality of the date hub is equal to 9. 

Now, consider the blue party hub, which originally has the biggest degree (along with the red party hub). Their star centrality values can be found by considering the induced star centered at the blue (red) party hub and the date hub itself (having a value of 6): adding any of the other neighbors only serves to decrease that value. Last, let us consider the green party hub. For that node, we can easily verify that its degree and star centrality match (and are both equal to 3). 
\end{ex}

\begin{figure}[htbp]
\centering
{\includegraphics[scale=0.75]{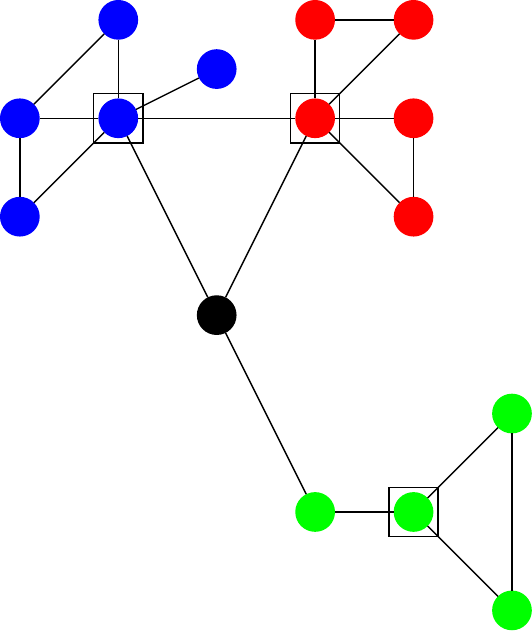}}
\caption{The PPIN of Figure \ref{partyVSdate} with the addition of an edge between the red and blue ``hub'' nodes. \label{partyVSdate2}}
\end{figure}

\begin{ex}
\label{partyVSdateExample2}
Consider again Figure \ref{partyVSdate}, however now assume that the red and blue ``hub'' nodes are connected with an edge (see Figure \ref{partyVSdate2} for a pictorial representation). Note that for the date hub we can no longer include both the neighboring blue and red nodes, as this does not lead to a feasible configuration. We also observe that this extra addition renders the node in black unnecessary for the connection of the red and blue complexes, and this will also become obvious when looking at the change in the star centrality metric. Seeing as now the blue and red hub nodes cannot both belong to the induced star centered at the center node, we will have to choose one or the other. This means that the star centrality metric for the center node in black is now equal to 6 (compared to the earlier value of 9). The blue and red hub nodes now both have a star centrality of 8. Finally, the green party hub is unaffected by the change. 
\end{ex}

From a biological perspective, we can use as an example the {\em YHL011C} protein from {\em Saccharomyces cerevisiae}. This is an essential protein used to synthesize PRPP, which is required for a mutlitude of important cell activities, including the 5-phosphoribose 1-diphosphate synthase necessary for tryptophan biosynthesis. This particular protein has a relatively small degree of 92 neighbors, when only considering protein interactions with a threshold equal to 60\% or more. Its degree alone would be enough to discard it from consideration as an essential protein, as it places it within the latter half of the degree ranking of all $6,418$ proteins. The same observation can be made for betweenness (equal to $0.00021$) and its eigenvector centrality ($0.00528$), with only closeness centrality being big enough to label it ``of importance'' (its closeness is equal to $0.404$). That said, its star centrality is equal to $2,749$, placing it in the top 10 of proteins when ranked for the size of the open neighborhood of the induced stars centered at them. 

A second example can be protein {\em YDL098C}, again from the {\em Saccharomyces cerevisiae} proteome, which also happens to be essential. This protein, also referred to as SNU23 is an important ingredient of the spliceosome: this would also be functionally predicted using network-based approaches in the past by \cite{li2010network}. Once more, though, when considering a PPIN generated by discarding all interactions below 60\%, its degree would be among the smallest in the proteome (it is equal to 31), ranking it in the bottom quarter. The same is also true here for its betweenness ($8.59\cdot10^{-7}$), eigenvector ($0.0008$) centralities, while its closeness centrality (0.305) places it in the bottom half of the rankings. On the contrary, its star centrality is equal to $461$, and this is enough to locate it in the top 1000 of most important proteins as far as this ranking is concerned.

\subsection{Complexity} 

\subsection{Complexity} 

In this subsection, we provide the computational complexity of the problem of detecting the node of maximum star degree centrality. We first give the decision version of the problem at hand in Definition \ref{SCDef}. 

\begin{dfn}[\textsc{Star Degree Centrality}]
\label{SCDef}
Given a graph $G(V, E)$ and an integer $k$, does there exist an induced star $S$ centered at any node $i\in V$ such that $|N(S)|\geq k$?
\end{dfn}

We proceed to derive the complexity of the problem using the well-known $\mathcal{NP}$-complete problem, \textsc{Independent Set}, whose decision version is provided in Definition \ref{ISDef}. 

\begin{dfn}[\textsc{Independent Set}]
\label{ISDef}
Given a graph $G(V, E)$ and an integer $k$, does there exist a set $S\subseteq V$ such that $|S|\geq k$ and for any two nodes $i,j\in S$, $(i,j)\notin E$?
\end{dfn}

We are now ready to prove that the problem we are tackling cannot admit a polynomial-time algorithm, under the assumption that $\mathcal{P}\neq\mathcal{NP}$. This is given in Theorem \ref{SCNPC}. 

\begin{thm}
\label{SCNPC}
	\textsc{Star Degree Centrality} is $\mathcal{NP}$-complete. 
\end{thm}

\proof{Proof.}
First of all, we verify that the problem is in $\mathcal{NP}$. Given a set of nodes $S\subseteq V$, we can verify that $S$ forms an induced star (one center with degree of $|S|-1$ and no edges between leaves), and that $\left|N(S)\right| \geq k$ in polynomial time.

Now, consider an instance of \textsc{Independent Set} $<G, k>$. We construct an instance of \textsc{Star Degree Centrality} $<\hat{G}, \ell>$ as follows. The vertex set and edge set of $\hat{G}$ are defined as: 

\begin{align*}
	&V[\hat G]=\hat V = V\cup \left\{s \right\} \cup \left\{ \cup_{i=1}^n  \{\cup_{j=1}^n s^{(i)}_j\} \right\} \cup \left\{\cup_{i=1}^{n^2} s^{(s)}_i \right\} \\
	&E[\hat G]=\hat E = E\cup \{\cup_{i=1}^n (s, i)\} \cup \{ \cup_{i=1}^n \{\cup_{j=1}^n (i, s^{(i)}_j)\}\}\cup\left\{ \cup_{i=1}^{n^2} (s, s^{(s)}_i) \right\}
\end{align*}

The above imply that graph $\hat G$ contains all nodes and edges from $G$. It also includes a newly added node $s$, that is connected by an edge to every node in $V$. Finally, $\hat{G}$ includes a total of $n=|V|$ nodes adjacent to every node $i\in V$ ($s_j^{(i)}$, for $j=1, \dots, n$), and a total of $n^2$ nodes ($s_j^{(s)}$, for $j=1, \dots, n^2$) that are only adjacent to $s$. Overall, the new graph has $2n^2+n+1$ nodes and $2n^2+m+n$ edges. Furthermore, let $\ell=n^2+k\cdot n$. An indicative, small example of the reduction from \textsc{Independent Set} to \textsc{Star Degree Centrality} can be found in Figures \ref{reductionAbove} and \ref{reductionBelow}.

\begin{figure}
\centering
{\includegraphics{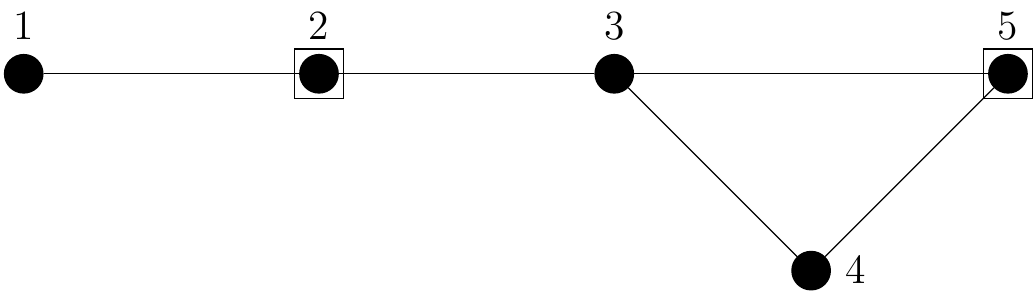}}
\caption{An example of the reduction for a graph $G$. This is the original graph $G(V,E)$, where nodes 2 and 5 form an \textsc{Independent Set} of size $k=2$ (one of the possible solutions). \label{reductionAbove}}
\end{figure}

\begin{figure}
\centering
{\includegraphics{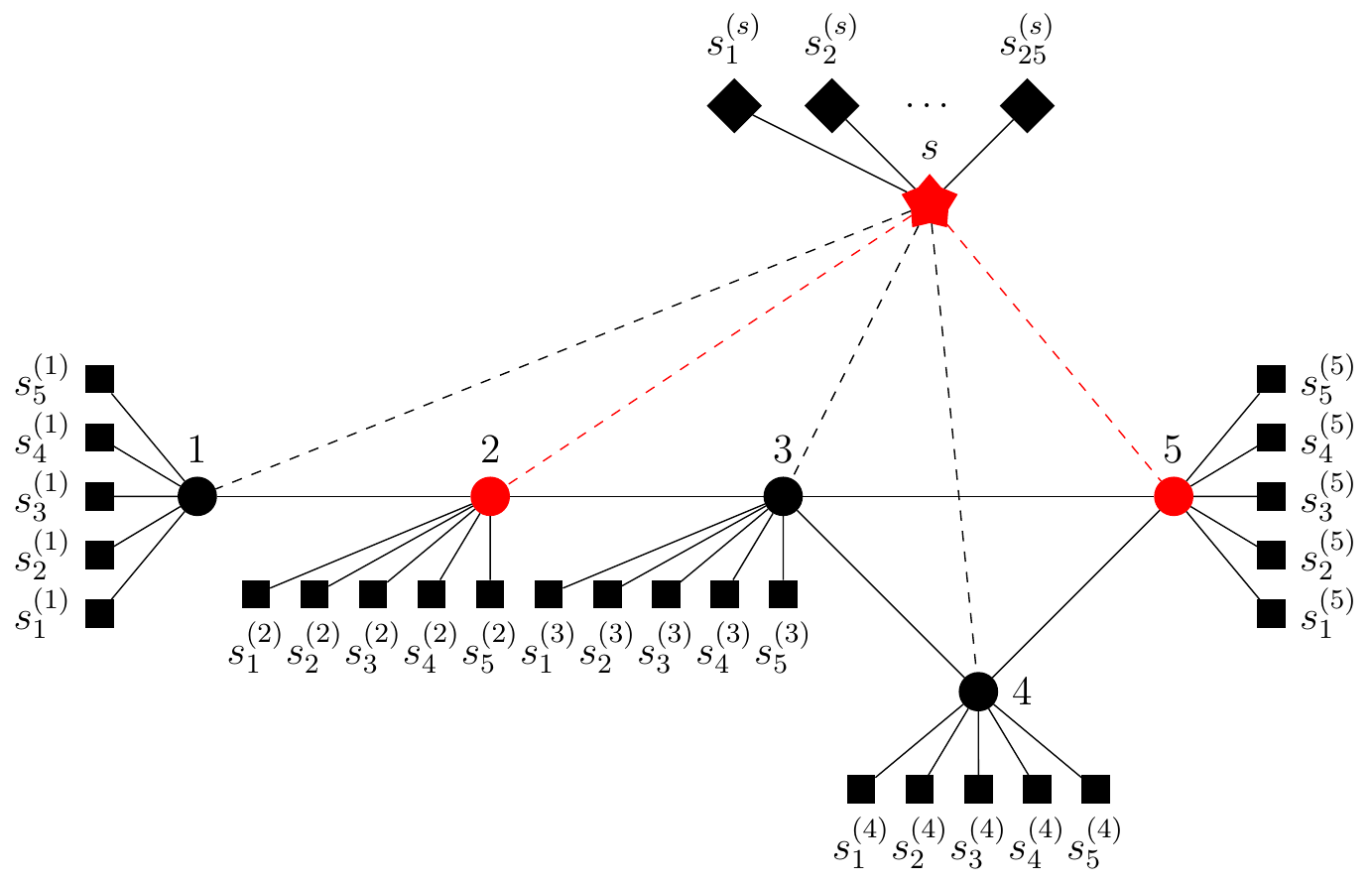}}
\caption{Here, we present an example of the gadget used to transform the instance of \textsc{Independent Set} $<G(V,E),k>$ for the graph of Figure \ref{reductionAbove} to an instance $<\hat{G}(\hat{V}, \hat{E}), \ell>$ of \textsc{Star Degree Centrality}. In the graph of Figure \ref{reductionAbove}, it is easy to see that nodes 2 and 5 form an \textsc{Independent Set} of size $k=2$, as they are not connected by an edge. In the graph below (Figure \ref{reductionBelow}, which presents $\hat{G}$), there exists an induced star $S\subseteq \hat{V}$ such that $|N(S)|\geq \ell=n^2+k\cdot n=25+10=35$. The star represents the newly added node, the squares are the $n$ nodes in $\hat G$ connecting to every node in $V$, while the diamonds the $n^2$ nodes connected to the star node. \label{reductionBelow}}
\end{figure}

First, let $S$ be an independent set of size $k$ in $G$. Then, consider the set of nodes $\hat S=S\cup \left\{ s\right\}$. We proceed to show that $\hat{S}$ forms an induced star, such that $|N(\hat{S})|\geq \ell$. This is true because, by construction, the nodes in $S$ are all adjacent to $s$; furthermore, no two nodes in $S$ are adjacent, as they form an independent set. Last, it is straightforward to see that $|N(\hat S)|\geq \ell= n^2+k\cdot n$. 

Now, assume that there exists no independent set of size $k$ in $G$. For a contradiction, we assume that there exists an induced star $\hat S\subseteq \hat V$ such that $|N(\hat S)|\geq \ell=k\cdot n+ n^2$. First of all, we note that $s\in \hat{S}$: if not, then there can be no star using nodes from $\hat V\setminus\left\{s\right\}$ with an open neighborhood of size at least equal to $n^2$. Further, there exist at least $k$ nodes from $V$ in the star: once more, if that is not the case, then $n^2\leq |N(\hat S)|< n^2+k\cdot n$. Last, observe that $\hat S$ is centered in $s$: assume for a contradiction that the star is instead centered at a node $i\in V$. Then, one of the two following cases has to hold: 

\begin{enumerate}[label=(\alph*)]
\item $\hat{S}$ contains $s$. This implies, by construction, that no other node $j\in V$ can be in the star, as $(s,j)\in \hat E$;
\item $\hat{S}$ contains at least $k$ nodes in $V$. This implies that $s$ cannot belong in $\hat S$, for the same reason as above. 
\end{enumerate}

In both cases, we observe that we reach a contradiction, hence $\hat{S}$ has to be an induced star centered at $s$. Since $\hat{S}$ forms an induced star, there exists no edge connecting any two nodes in $\hat{S}\setminus\left\{s\right\}$, rendering $\hat{S}\setminus\left\{s\right\}$ an independent set in $G$. Finally, $\hat{S}$ contains $k$ nodes in $V$, hence $|\hat{S}\setminus\left\{s\right\}|=k$, which implies that an independent set of size $k$ exists in $G$. This contradiction finishes the proof. 
\endproof

\subsection{Extensions}

It can also be shown that the star centrality function is submodular. 

\begin{thm}
The function $f(S)=\{|N(S)|: S \textit{ forms an induced star}\}$ is submodular. 
\end{thm}

\proof{Proof.}
Let $S_1, S_2$ be two induced stars such that $S_1\subseteq S_2$. Also, consider a node $u\in V\setminus S_2$. Then, we have that: 

\begin{align*}
& f(S_1 \cup \left\{u\right\}) - f(S_1) = -1 + |N(u)|- |N(u)\cap N[S_1]| \\
& f(S_2 \cup \left\{u\right\}) - f(S_2) = -1 + |N(u)|- |N(u)\cap N[S_2]| 
\end{align*}

It is clear that $|N(u)\cap N[S_1] \leq |N(u)\cap N[S_2]|$, as $N[S_1]\subseteq N[S_2]$, and hence, $f(S_1 \cup \left\{u\right\}) - f(S_1) \geq f(S_2 \cup \left\{u\right\}) - f(S_2)$. 
\endproof

Unfortunately, though, the star centrality function is not monotone; consider a node with no neighbors other than to a designated center. Then, that particular node can be added as a leaf to the star, however it would only serve to decrease its open neighborhood size by 1. For better exposition, a counterexample is presented in Figure \ref{monotoneCounter}. This implies that we cannot easily use a simple greedy approach to approximate the optimal solution. We do though provide a different greedy mechanism study in a subsequent section.

\begin{figure}
\centering
{\includegraphics{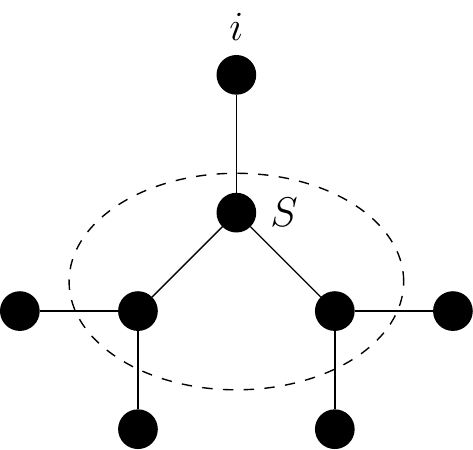}}
\caption{A counterexample of the monotonicity of the star centrality function. As can be easily seen, the open neighborhood size of the star $S$ is decreased by 1 when considering the star $S\cup\{i\}$. \label{monotoneCounter}}
\end{figure}

\section{Mathematical Formulation and Approximation Algorithms}
\label{model}

In this section, we provide a mathematical formulation for our problem, followed by two approximation algorithms. First, let us define the following decision variables: 

\begin{align*}
x_i = & \left\{\begin{tabular}{ll} 1, & if node $i\in V$ is the center of the star \\ 0, & otherwise. \end{tabular}\right. \\
y_i = & \left\{\begin{tabular}{ll} 1, & if node $i\in V$ \textit{ is in the star }\\ 0, & otherwise. \end{tabular}\right. \\
z_i = & \left\{\begin{tabular}{ll} 1, & if node $i\in V$ is adjacent to a node in the star \\ 0, & otherwise. \end{tabular}\right. \\
\end{align*}

\subsection{Mathematical Formulation}

The integer programming formulation for detecting the induced star of maximum degree centrality is presented in \eqref{form1}--\eqref{form6}. 

\begin{align}
\label{form1} 
\text{IP:}~\max &~ \sum\limits_{i\in V} z_i \\
\label{form2}
s.t. &~ y_i + z_i \leq 1, & \forall i\in V \\
\label{form3}
&~ z_i \leq \sum\limits_{j\in N(i)} y_j, & \forall i\in V \\
\label{form4}
&~ y_i \leq \sum\limits_{j\in N[i]} x_j, & \forall i\in V \\
\label{form5}
&~ y_i+y_j \leq 1+x_i+x_j, & \forall (i,j)\in E \\
\label{form6}
&~ \sum\limits_{i\in V} x_i = 1, \\
\label{form7}
&~ x_i, y_i, z_i \in\{0,1\}, & \forall i \in V.
\end{align}

Clearly, our objective is to maximize the size of the open neighborhood of the star, as shown in \eqref{form1}. Then, \eqref{form2} ensures that no node is allowed to be both in the star and in its open neighborhood. Constraint families \eqref{form3} and \eqref{form4} are similar in nature and enforce which nodes are adjacent to the star, and which nodes are adjacent to the center and, as such, can be considered for addition to the star. Moreover, no two leafs are allowed to be connected, as per constraint \eqref{form5}. Last, we are only looking for one star, enforced with \eqref{form6}, and all of our decision variables are binary. 

We can also consider the problem of detecting the star centrality of a given node $u\in V$, as shown in \eqref{form21}--\eqref{form27}. 

\begin{align}
\label{form21} 
\text{IP($u$):}~\max &~ \sum\limits_{i\in V} z_i \\
\label{form22}
s.t. &~ y_i + z_i \leq 1, & \forall i\in V \\
\label{form23}
&~ y_i \leq a_{iu}, & \forall i\in V\setminus \{u\} \\ 
\label{form25}
&~ z_i \leq \sum\limits_{j: (i,j)\in E} y_i, & \forall i\in V \\
\label{form26}
&~ y_i + y_j \leq 1, & \forall (i,j)\in E: i\neq u, j\neq u \\
\label{form24}
&~ y_u=1 \\
\label{form27}
&~ y_i, z_i \in\{0,1\}, & \forall i \in V.
\end{align}

The objective function, given at equation \eqref{form21}, as well as the constraint families in \eqref{form22}, \eqref{form25} and the variable restrictions in \eqref{form27} are identical to the previous model. However, note that we no longer need to consider a decision variable for the center of the star, as it is known to be node $u\in V$. Hence, we can add constraints \eqref{form23} that only consider the nodes that are adjacent to $u$ as candidates to be in the star, and modify constraint \eqref{form26} to only consider the connections that do not include the star center. As a reminder, $a_{ij}$ is the adjacency matrix entry that represents the connection between nodes $i$ and $j$. Last, constraint \eqref{form24} will force node $u$ (the center) to be part of the induced star. 

This last integer program, as shown in \eqref{form21}--\eqref{form27}, can be used to calculate the star centrality of each and every one of the nodes in the network. In our numerical experiments (presented in Section \ref{results}), the optimal objective function value of this integer program is the {\em star centrality} that is then compared to other, nodal centrality metrics.

\subsection{Greedy algorithms}

As shown earlier, we cannot unfortunately claim monotonicity for the star centrality function. Hence, deriving an approximation ratio from simply applying a greedy algorithm scheme is not straightforward. However, we can still show that the greedy algorithm, presented in Algorithm \ref{SimpleGreedy} has an approximation guarantee of $O(\Delta)$, where $\Delta$ is the maximum degree in the network. First, let us introduce for simplicity a function $f_1(S, k)$ to capture the ``gain" of adding a node $k$ to a star $S$, assuming of course that $S\cup\{k\}$ remains an induced star. 

\begin{align*}
f_1(S, k)= |N(S\cup\{k\})| - |N(S)|.
\end{align*}

We note that for this function we have that $f_1(S,k)\geq-1$. This follows from the fact that for any node $k$ such that $S\cup\{k\}$ forms an induced star, we have that $k\in N(S)$: hence, in the worst case, adding $k$ to $S$ decreases its open neighborhood size by 1 (as $k$ is no longer adjacent to the star, but instead is now part of it). Figure \ref{monotoneCounter} shows this worst case behavior, as adding node $i$ to star $S$ only serves to decrease its open neighborhood size by 1. 

\begin{align*}
f_1(S, k)= |N(S\cup\{k\})| - |N(S)|.
\end{align*}

\begin{algorithm}[t]
    \SetKwInOut{Input}{Input}
    \SetKwInOut{Output}{Output}

    \underline{function SimpleGreedy} $(i)$\;
    \Input{A node $i\in V$}
    \Output{An induced star $S$ centered at $i$}
    $candidates\leftarrow N(i)$\;
    $S\leftarrow\{i\}$\;
    \While{$candidates\neq \emptyset$}{
    	\For{$k\in candidates$}{
    		\If{$f_1(S, k)<=0$}{
			$candidates\leftarrow candidates\setminus\{k\}$\;
			}
    	}
	\If{$candidates\neq \emptyset$}{
  		$j\leftarrow\arg\max\limits_{k}\{f_1(S,k): k\in candidates\}$\;
		$S\leftarrow S\cup\{j\}$\;
		$candidates\leftarrow candidates\setminus\{N[j]\}$\
	}
  }
  \Return $S$
	\caption{Simple Greedy. \label{SimpleGreedy}}
\end{algorithm}

\begin{thm}
	\label{thmSimpleGreedyRatio}
	Let $i\in V$, with a degree of $\delta$, be the node whose star centrality we are interested in finding. Then, the simple greedy algorithm has an an approximation ratio of $O(\delta)$.
\end{thm}

\begin{proof}
At each iteration of the while loop, the greedy algorithm looks at the candidate nodes (set $\{j\in N(i)\setminus S: (k,j)\notin E, \forall k\in S\}$), and selects to add the one that is adjacent to the maximum number of not already covered nodes. In the worst case, the greedy algorithm terminates after the first iteration, and that only happens when the greedily selected node $u\in N(i)$ which adds $\alpha=|N(u)\setminus N[i]|$ is connected to every other node in $N(i)$. Let $OPT$ be the optimal value and $z_{greedy}$ the value obtained by applying the simple greedy approach. Then, we have that

\begin{align}
\label{optBound}
OPT &\leq (\delta-1)\cdot (\alpha-1) + 1 \leq (\delta-1)\cdot \alpha \\
\label{greedyBound}
z_{greedy} &\geq \alpha+\delta-1 \geq \alpha
\end{align}

From \eqref{optBound} and \eqref{greedyBound}, we obtain the approximation guarantee as

\begin{align}
\frac{OPT}{z_{greedy}} \leq \frac{(\delta-1)\cdot\alpha}{\alpha}=\delta-1 = O(\delta). 
\end{align}
\hfill
\end{proof}

\begin{figure}
\centering
{\includegraphics{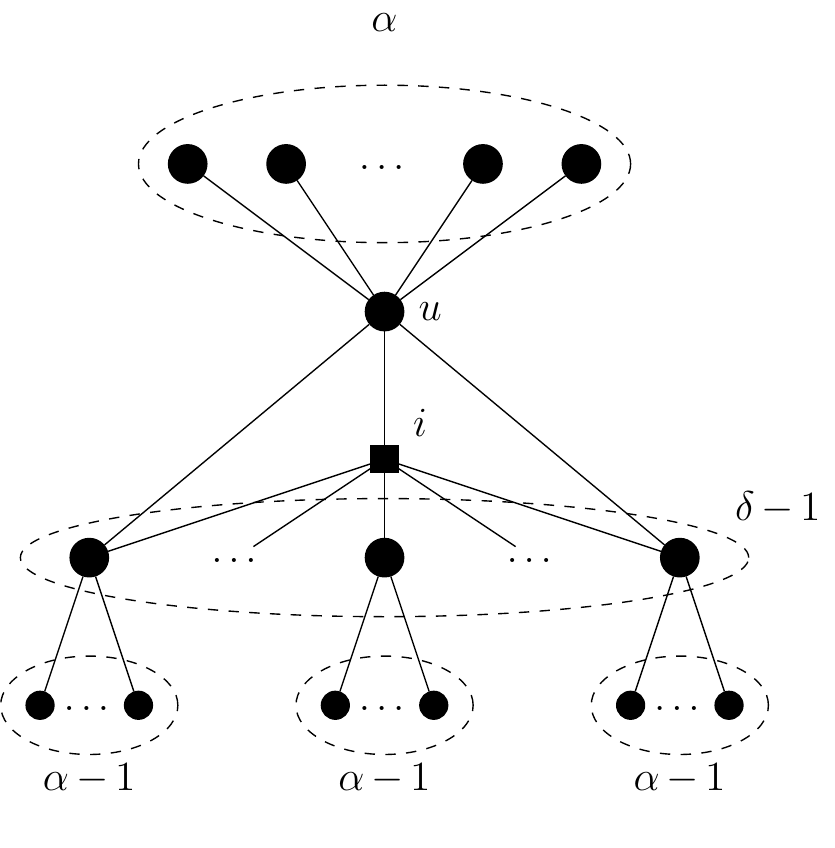}}
\caption{An example of the worst-case behavior guarantee of the Simple Greedy approach. In this case adding $u$ to the star centered at $i$ results in a star centrality of $\alpha+\delta-1$, while adding every other neighbor of $i$ to the star would result in $(\delta-1)\cdot (\alpha-1)+1$. \label{SimpleBound}}
\end{figure}

Figure \ref{SimpleBound} shows an example of the worst-case performance. Let us now propose a different greedy-based heuristic algorithm and show its approximation ratio. Let $S^i$ be again an induced star centered at $i$ and define function $f_2(S^i, k)$ as:

\begin{align*}
f_2(S^i, k)=\sum\limits_{j:(i,j)\in E, (k,j)\in E} |\left(N(S^i \cup\{j\}\right)\setminus N(S^i)|.
\end{align*}

This function captures the potential increase in the size of the open neighborhood that we would be losing since nodes $j$ and $k$ cannot belong to the star simultaneously. Note here that $f_2(S^i, k)=0$ implies either that node $k$ is connected to no other potential leaf of the star, or that all other candidates connected to $k$ add no uncovered nodes to the star. Now, consider the greedy approach shown in Algorithm \ref{RatioGreedy}. We show its approximation ratio in Theorem \ref{thmApprox}; to do that, we first provide two lemmata. 

\begin{algorithm}
    \SetKwInOut{Input}{Input}
    \SetKwInOut{Output}{Output}

    \underline{function RatioGreedy} $(i)$\;
    \Input{A node $i\in V$}
    \Output{An induced star $S$ centered at $i$}
    $S^i\leftarrow\{i\}$\;
    $candidates_1\leftarrow\{k\in N(i): f_2(S^i, k)=0\}$\;
    $candidates_2\leftarrow N(i)\setminus candidates_1$\;
    \While{$candidates_1\neq \emptyset$ or $candidates_2\neq \emptyset$}{
    	\eIf{$candidates_1\neq \emptyset$}{
		$j\leftarrow\arg\max\limits_{k}\{f_1(S^i,k): k\in candidates_1, f_1(S^i,k)>0\}$\;
		$S^i\leftarrow S^i\cup\{j\}$\;
		$candidates_1\leftarrow candidates_1\setminus\{j\}$\;
			
	}
	{
		$j\leftarrow\arg\max\limits_{k}\{\frac{f_1(S^i,k)}{f_2(S^i,k)}: k\in candidates_2, f_1(S^i,k)>0\}$\;
		$S^i\leftarrow S^i\cup\{j\}$\;
		$candidates_2\leftarrow candidates_2\setminus N[j]$\;
		
	}
	\For{$k\in candidates_2$}{
			\If{$f_2(S^i, k)=0$}{
				$candidates_2\leftarrow candidates_2\setminus\{k\}$\;
				$candidates_1\leftarrow candidates_1\cup \{k\}$\;
			}
		}  	
  }
  	\Return $S$
	\caption{Ratio-based Greedy. \label{RatioGreedy}}
\end{algorithm}

\begin{lem}
Let $i\in V$, with a degree of $\delta$, be the node whose star centrality we are interested in finding. Further, assume that for all nodes $k\in N(i)$, we have that $f_2(S^i, k)=0$, that is there exists no connection between any two of them. Then, greedily selecting the node with maximum $f_1(S^i, k)$ has an approximation ratio of $O(\ln \delta)$.
\end{lem}

\begin{proof}

It can be seen that the above setup results in greedily solving a set cover problem with $\delta$ sets. The universe of elements to be covered is all nodes reachable within 1 or 2 hops from $i$, $N^2(i)$. Each set consists of the neighbors of $i$ and their neighbors which belong to $N^2(i)$, that is $C_j=\{j, N(j)\cap N^2(i)\}, \forall j\in N(i)$. Since applying the greedy algorithm results in an $O(\ln n)$ approximation for the set cover and we have at most $\delta$ candidate nodes/sets, all of which can be selected at any point, as there exist no connections between them, the greedy algorithm would result in an  $O(\ln \delta)$ approximation ratio, as far as the number of nodes added to the star is concerned. Let $OPT_{SC}$ represent the optimal solution to the set cover problem above and $z_{SC}$ the solution using the greedy algorithm. We then have that:

\begin{align}
	\label{RatioOPT}
	OPT &= |N^2(i)|-OPT_{SC} \\
	\label{RatioGREEDY1}
	z_{greedy} &= |N^2(i)| - z_{SC} \geq |N^2(i)|-\ln\delta \cdot OPT_{SC}
\end{align}

Combining \eqref{RatioOPT} and \eqref{RatioGREEDY1}, we obtain that: 

\begin{align}
	\label{approxRatioBased}
	\frac{OPT}{z_{greedy}} &\geq \frac{|N^2(i)|-OPT_{SC}}{|N^2(i)|-\ln\delta \cdot OPT_{SC}} \geq \frac{1}{\ln\delta} \implies z_{greedy}\leq \ln\delta\cdot OPT.
\end{align}

The last inequality proves the Lemma. \hfill 
\end{proof}

\begin{lem}
Let $i\in V$, with a degree of $\delta$, be the node whose star centrality we are interested in finding. Further, assume that for all nodes $k\in N(i)$, we have that $f_2(S^i, k)>0$, that is each node is connected to at least one other in $N(i)$. Then, greedily selecting the node with maximum $\frac{f_1(S^i, k)}{f_2(S^i, k)}$ has an approximation ratio of $O(\sqrt \delta)$.
\end{lem}

\begin{proof}
Similarly to the case in Theorem \ref{thmSimpleGreedyRatio}, the worst case behavior is observed when the algorithm terminates after adding only one node in the star. This can happen when the selected node is indeed adjacent to all other nodes in $N(i)$. Let $\beta_j$ be the nodes adjacent to $j$ that are not already in $S$ or covered by $S$. Furthermore, let node $u$ be connected to all other candidate nodes. We then have that:

$$a_u=\frac{f_1(S, u)}{f_2(S,u)} = \frac{\beta_u}{\sum\limits_{k\in N(i), k\neq u} \beta_k},$$

\noindent while, for the remaining nodes, $j\neq u$, we would have that: 

$$a_j \leq \frac{\beta_j}{\beta_u}.$$

In the worst case, the remaining nodes can all be part of the same star (i.e., there exist no connections between them). Hence, to select node $u$ using the ratio-based greedy approach we must have $a_u \geq a_j$, for all $j$, and assuming $v$ is the nodes with maximum ratio when excluding $u$, we have that $a_u\geq a_v$. This implies: 

\begin{align}
\nonumber
	a_u \geq a_v &\implies \frac{\beta_u}{\sum\limits_{k\in N(i), k\neq u} \beta_k} \geq \frac{\beta_v}{\beta_u} \implies \frac{\beta_u}{(\delta-1)\cdot \beta_v} \geq \frac{\beta_v}{\beta_u}  \\
	&\implies \beta_u^2 \geq (\delta-1)\cdot \beta_v^2 \implies \beta_v \leq \frac{\beta_u}{\sqrt{\delta-1}}.
\end{align}

Hence, in the worst case, the greedy algorithm results in a solution of $\beta_u+\delta-1$, while the optimal solution can be as big as $(\delta-1)\cdot \frac{\beta_u}{\sqrt{\delta-1}} + 1$. We finally get: 

\begin{align}
\frac{OPT}{z_{greedy}} \leq \frac{(\delta-1)\cdot \frac{\beta_u}{\sqrt{\delta-1}} + 1}{\beta_u+\delta-1} \leq \frac{\sqrt{\delta-1}\cdot\beta_u}{\beta_u} = O(\sqrt{\delta}). 
\end{align}
\hfill 
\end{proof}

\begin{thm}
\label{thmApprox}
Let $i\in V$, with a degree of $\delta$, be the node whose star centrality we are interested in finding. Then, the ratio-based greedy algorithm has an approximation ratio of $O(\sqrt\delta)$. 
\end{thm}

\begin{proof}
The algorithm is divided into two phases: in the first phase, the node with the maximum ratio is selected, while in the latter one, we choose the node with the maximum number of uncovered neighbors. 

Let $OPT_1$ and $OPT_2$ represent the optimal solutions obtained from each phase. Then, $OPT\leq OPT_1+OPT_2$. Similarly, let $z_1$ and $z_2$ be the solutions obtained from each phase of the greedy algorithm; it is easy to see that $z_{greedy}=z_1+z_2$. From the previous lemmata, we have that:

\begin{align}
OPT_1&\leq O(\sqrt{\delta})\cdot z_1 \\
OPT_2&\leq O(\ln{\delta})\cdot z_2.
\end{align}

Combining, we get that 

\begin{align}
\nonumber\frac{OPT}{z_{greedy}} &\leq \frac{OPT_1+OPT_2}{z_{greedy}} \leq \frac{O(\sqrt{\delta})\cdot z_1 + O(\ln{\delta})\cdot z_2}{z_1+z_2} \leq \\  &\leq \frac{O(\sqrt{\delta})\cdot (z_1+z_2)}{z_1+z_2} = O(\sqrt{\delta}).
\end{align}

\hfill

\end{proof}

\section{Computational results}
\label{results}

In this section, we present our experimental setup, the data used, and analyze and interpret the results obtained. Our goal is to portray how star centrality behaves and performs when put to the test against other popular centrality metrics in PPIN analysis. 

\subsection{Experimental setup}

All numerical experiments were performed on a quad-core Intel i7 at 2.8 GHz with 16 GB of RAM. The codes were written in Python and C++ and, where needed, the Gurobi 6.50 solver \citep{gurobi} was used to solve the optimization problems. Data on protein interactions for different organisms was obtained by STRING v. 10.0 \citep{szklarczyk2014string}. More specifically, we used the datasets of \emph{Saccharomyces cerevisiae} (yeast), \emph{Helicobacter pylori}, and \emph{Staphylococcus aureus} (presented in this section), and \emph{Salmonella enterica CT18},  \emph{Caenorhabditis elegans} (presented in the Appendix). Essentiality for proteins was found using the databases for the above organisms as curated in DEG 10 \citep{luo2013deg}.

We performed two experiments. In the first one, which is presented in subsection \ref{analysis1}, the PPINs were created as follows. For each protein in the database, a node was created and was connected to all other proteins-nodes that they shared an interaction. Then, all interactions-edges with an interaction score that was below a threshold were removed. Seeing as the maximum interaction score was 1000, the threshold scores selected for presentation in this study were 600 (60\% interaction score), 700 (70\% interaction score), and 800 (80\% interaction score). In this fashion, we were able to create three networks per organism where all known centrality metrics can be captured given the computational power. The networks were further broken down into their connected components with each component being independently analyzed, without loss of generality. 

For the second experiment, discussed in subsection \ref{analysis2}, twenty different networks were created for three of the previous organisms. Each network was generated by randomly adding every protein-protein interaction present in the datasets with a probability equal to the interaction score divided by the maximum interaction present (1000). As an example, a protein-protein interaction with a score of 550 in the database would appear in the generated network with a probability of 0.55. The goal of this second experiment is to measure how many times a protein appears among the top ranked (per a specific metric) in the generated networks. 

All nodal metrics of centrality (\emph{degree}, \emph{closeness}, \emph{betweenness}, \emph{eigenvector}) were computed with a Python implementation, using NetworkX 1.9 \citep{hagberg-2008-exploring}. On the other hand, \emph{star centrality} calculations were performed on the same networks with a C++ implementation. For all networks an exact solution was found for every node; however, we also obtained an approximate solution using the greedy techniques proposed (as described in the previous section).

Finally, in subsection \ref{greedyAnalysis}, we contrast the performance of the approximation algorithms to the exact solution, both as far as time and solution quality are concerned. For this analysis we employ the same networks that were used for the first experiment.

\subsection{Analysis of top ranked proteins per metric}
\label{analysis1}

After obtaining all the metrics for the PPIN in consideration, we calculated the ratio of essential proteins found in the top $k$ and the bottom $k$ proteins (as ranked by each centrality metric). The bounds for each analysis are shown in Table \ref{bounds} and were based on the total number of essential proteins present in the PPIN for each organism. As an example, for \emph{Helicobacter pylori}, the number of essential proteins found in the PPIN was 435, and hence the top 500 proteins were investigated. Observe that the closer a metric gets to 100\%, the more accurately it detects essential proteins. 

\begin{table}
\tiny
\caption{Details of the PPINs and the bounds selected for each organism analysis. \label{bounds}}
{\begin{tabular}{l|c|rrr|rrrr}
 		&  & \multicolumn{3}{c|}{$|E|$} & \multicolumn{1}{c}{Essential} &  &  \\
Organism		&$|V|$	    & 	600 & 700 & 800 	&  \multicolumn{1}{c}{Proteins} & \multicolumn{1}{c}{Top $k$} & \multicolumn{1}{c}{Bottom $k$}  \\
\hline
\emph{Saccharomyces cerevisiae} & 6,418 &  179,317  & 137,304 & 99,705
& 1,221 & 1,000 & 5,00 \\
\emph{Helicobacter pylori} &  1,570 &  17,792    & 12,822 & 7,859
& 431 & 500 & 500 \\
\emph{Staphylococcus aureus} & 2,853 & 16,857   & 11,996 & 8,530
& 314 & 400 & 400 \\
\emph{Salmonella enterica CT18} & 4,529 & 40,165    & 27,649 &18,547
& 543 & 500 & 500 \\
\emph{Caenorhabditis elegans} & 15,830   & 322,294   & 202,834 & 129,250
& 492 & 500 & 500 \\
\end{tabular}}
\end{table}

For each organism then, we provide three Figures: one representing the performance over the top $k$ proteins, a second one over the bottom $k$ proteins, along with a Receiver Operating Characteristic curve (ROC curve). Note that for the first representation, the higher the ratio is then the better that metric is said to perform. The opposite is true for the second representation as a metric is said to perform better if the ratio is smaller. Last, for the third representation, the higher the area under the curve (AUC), the better the metric is said to perform. In this section, we only provide the figures corresponding to the \emph{Saccharomyces cerevisiae}, \emph{Helicobacter pylori}, and \emph{Staphylococcus aureus} organisms and the first of the three thresholds selected (60\%): the other two thresholds (70\% and 80\%) and all remaining organisms are given in the Appendix. 

First, let us consider Figures \ref{SCerevisiaeTop} and \ref{SCerevisiaeBottom} that show our results for \emph{Saccharomyces cerevisiae}: the star centrality metric is outperforming every other considered nodal centrality metric with a final performance of having $50.1\%$ of all essential proteins within the top 1000. Note that the maximum that could be achieved here would be 81.9\%, making the effective detection rate equal to 61.17\%. On the contrary, the other centrality metrics are almost indistinguishable and achieve a final performance of $22.11\%$, $22.03\%$, $22.52\%$, and $23.01\%$ for degree, closeness, betweenness, and eigenvector centrality, respectively. In the bottom 500 proteins, star centrality is still performing better, albeit less so than earlier, achieving a final score of $9.17\%$, as compared to the final scores of $10.4\%$, $10.24\%$, $9.91\%$, and $9.91\%$ for the other centrality metrics. We see a similar behavior in the ROC curve, shown in Figure \ref{ROC1} (left), where the area under the curve for star centrality is 0.766 (compared to 0.672 for degree, 0.548 for betweenness, 0.669 for closeness, and 0.682 for eigenvector).

\begin{figure}
	\begin{minipage}{0.45\textwidth}
  	\tiny
	{\includegraphics{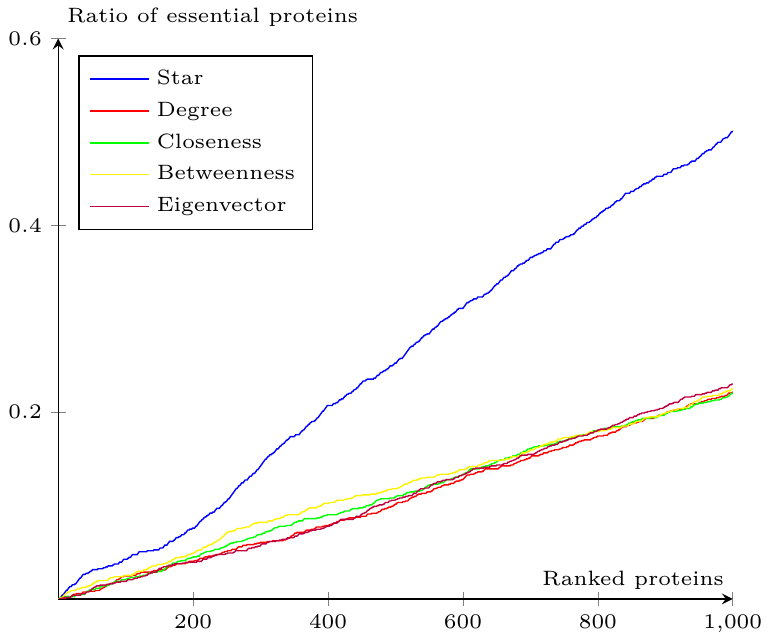}}
	\caption{The ratio of essential proteins detected in the ranked \emph{top k} proteins according to each metric for the \emph{Saccharomyces cerevisiae} organism (yeast) when a threshold of 60\% was used. \label{SCerevisiaeTop}}
  \end{minipage} ~~~%
  \begin{minipage}{0.45\textwidth}
\tiny
	{\includegraphics{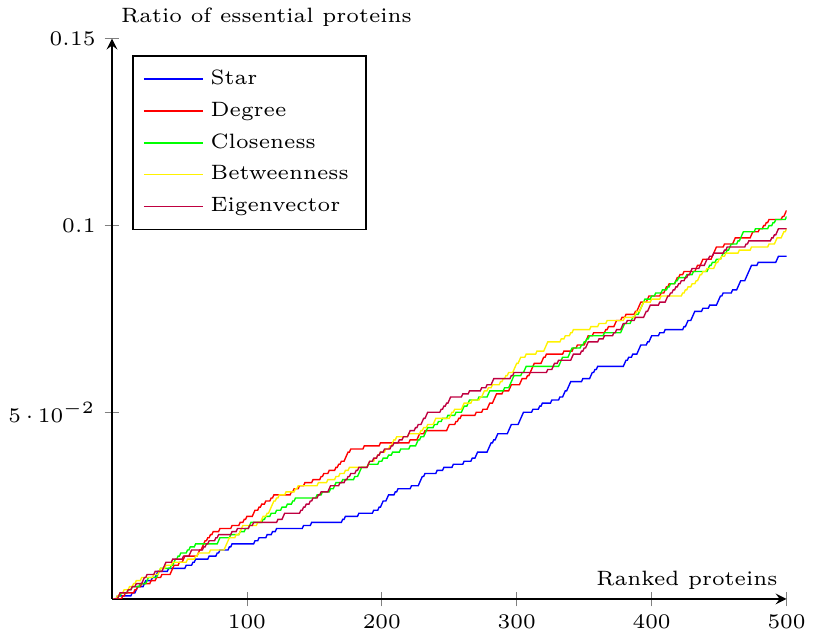}}
	\caption{The ratio of essential proteins detected in the ranked \emph{bottom k} proteins according to each metric for the \emph{Saccharomyces cerevisiae} organism (yeast) when a threshold of 60\% was used. \label{SCerevisiaeBottom}}
  \end{minipage}
\end{figure}

In the case of the \emph{Helicobacter pylori} organism, shown in Figures \ref{HPylorisTop} and \ref{HPylorisBottom}, and Figure \ref{ROC1} (center), the situation is similar. Star centrality achieves a final score of detecting $55.65\%$ within the top 500 proteins, as opposed to $38.6\%$ for degree centrality, $47.63\%$ for closeness centrality, $34.09\%$ for betweenness centrality, and $40.63\%$ for eigenvector centrality. Considering the performance over the least well ranked proteins, it is easier to see that star centrality is best at not ranking highly non-essential proteins, achieving a final score of $19.49\%$, while the scores for the other centrality metrics are significantly higher at $37.82\%$, $32.51\%$, $41.31\%$, and $32.51\%$. The area under the ROC curve ends up being 0.753, outperforming every other centrality metric in this study.

\begin{figure}
\centering
	\begin{minipage}{0.45\textwidth}
  	\tiny
	{\includegraphics{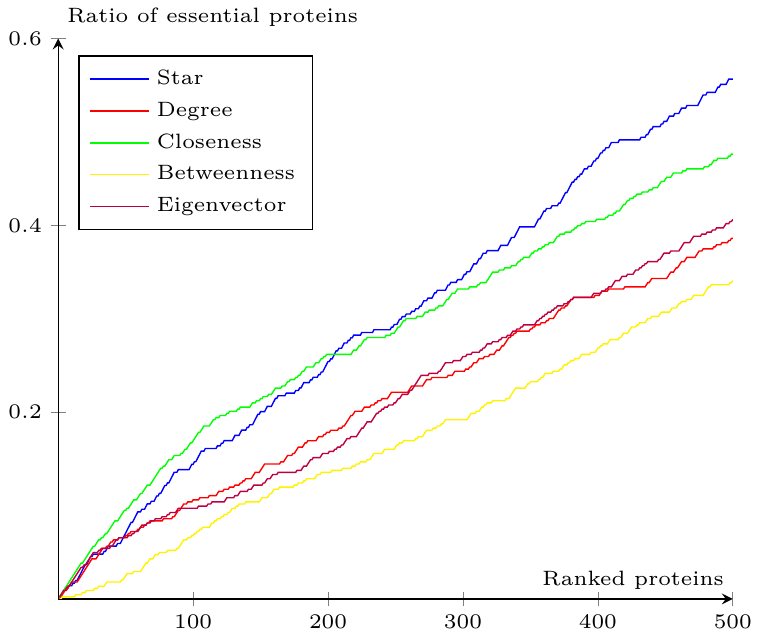}}
	\caption{The ratio of essential proteins detected in the ranked \emph{top k} proteins according to each metric for the \emph{Helicobacter pylori} organism when a threshold of 60\% was used. \label{HPylorisTop}}
  \end{minipage} ~~~%
  \begin{minipage}{0.45\textwidth}
\tiny
	{\includegraphics{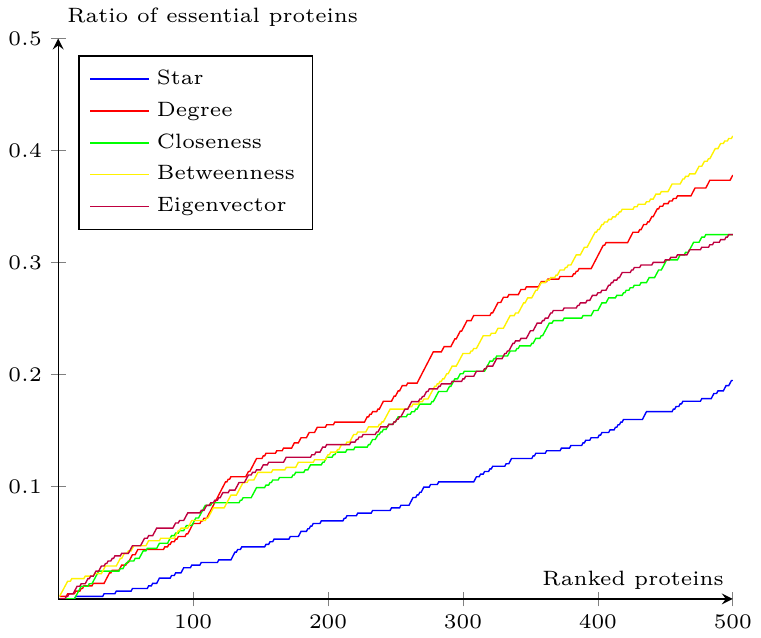}}
	\caption{The ratio of essential proteins detected in the ranked \emph{bottom k} proteins according to each metric for the \emph{Helicobacter pylori} organism when a threshold of 60\% was used. \label{HPylorisBottom}}
  \end{minipage}
\end{figure}

Continuing with the results in the \emph{Staphylococcus aureus} organism, presented in Figures \ref{StaphTop}, \ref{StaphBottom}, and \ref{ROC1} (right), the same pattern is again seen. Star centrality consistently outperforms the other nodal metrics, and its accuracy is much higher at any given step in the analysis. Overall, the final star centrality score is $65.61\%$, which easily outperforms the final scores of the other centrality metrics, 40.21\%, 38.14\%, 39.18\%, and 34.02\%, respectively. Similarly, when considering the bottom 400 proteins, we obtain a final score of 7.96\% for star centrality, as compared to the very high 38.14\%, 45.36\%, 37.11\%, and 29.90\% for the remaining centrality metrics. The area under the curve is as big as 0.867 with the eigenvector centrality behaving well with an area of 0.788. 

\begin{figure}
\centering
	\begin{minipage}{0.45\textwidth}
  	\tiny
	{\includegraphics{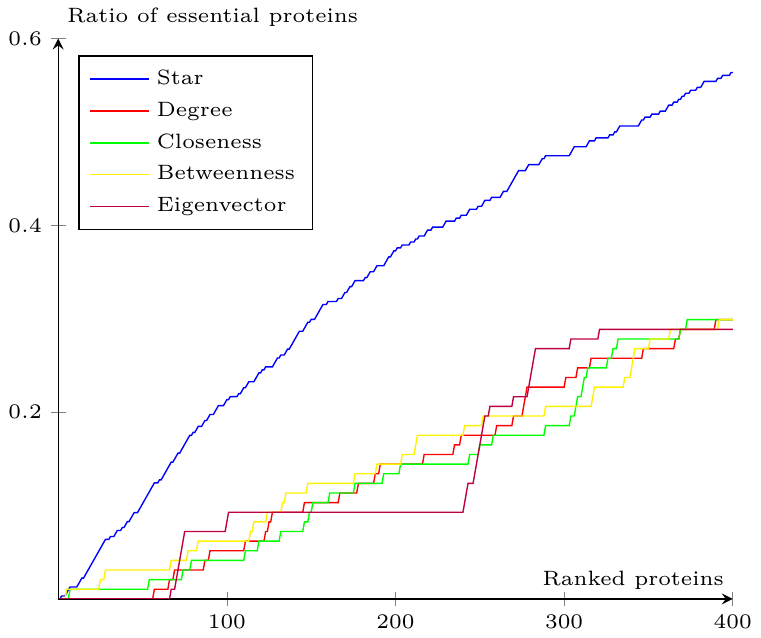}}
	\caption{The ratio of essential proteins detected in the ranked \emph{top k} proteins according to each metric for the \emph{Staphylococcus aureus} organism when a threshold of 60\% was used. \label{StaphTop}}
  \end{minipage} ~~~%
  \begin{minipage}{0.45\textwidth}
\tiny
	{\includegraphics{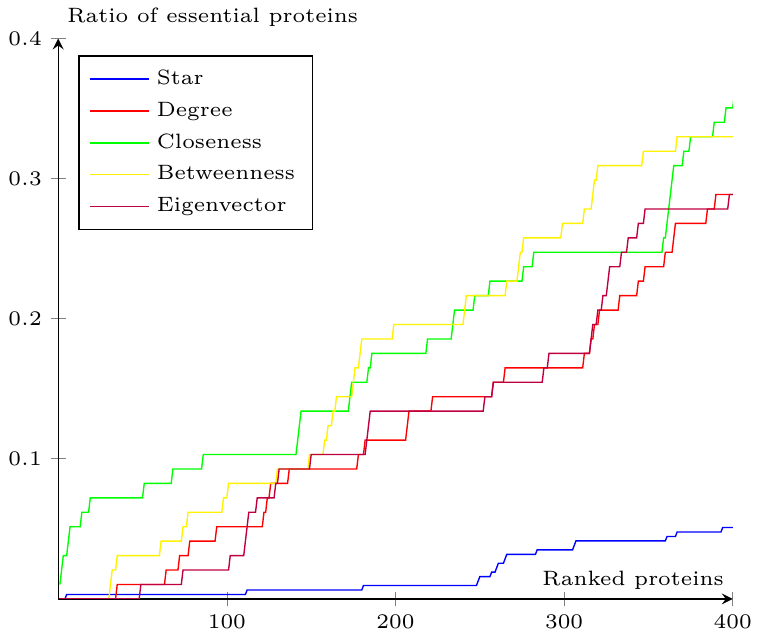}}
	\caption{The ratio of essential proteins detected in the ranked \emph{bottom k} proteins according to each metric for the \emph{Staphylococcus aureus} organism when a threshold of 60\% was used. \label{StaphBottom}}
  \end{minipage}
\end{figure}



\begin{figure}
	\tiny
	\centering
	{\includegraphics[width=0.33\textwidth]{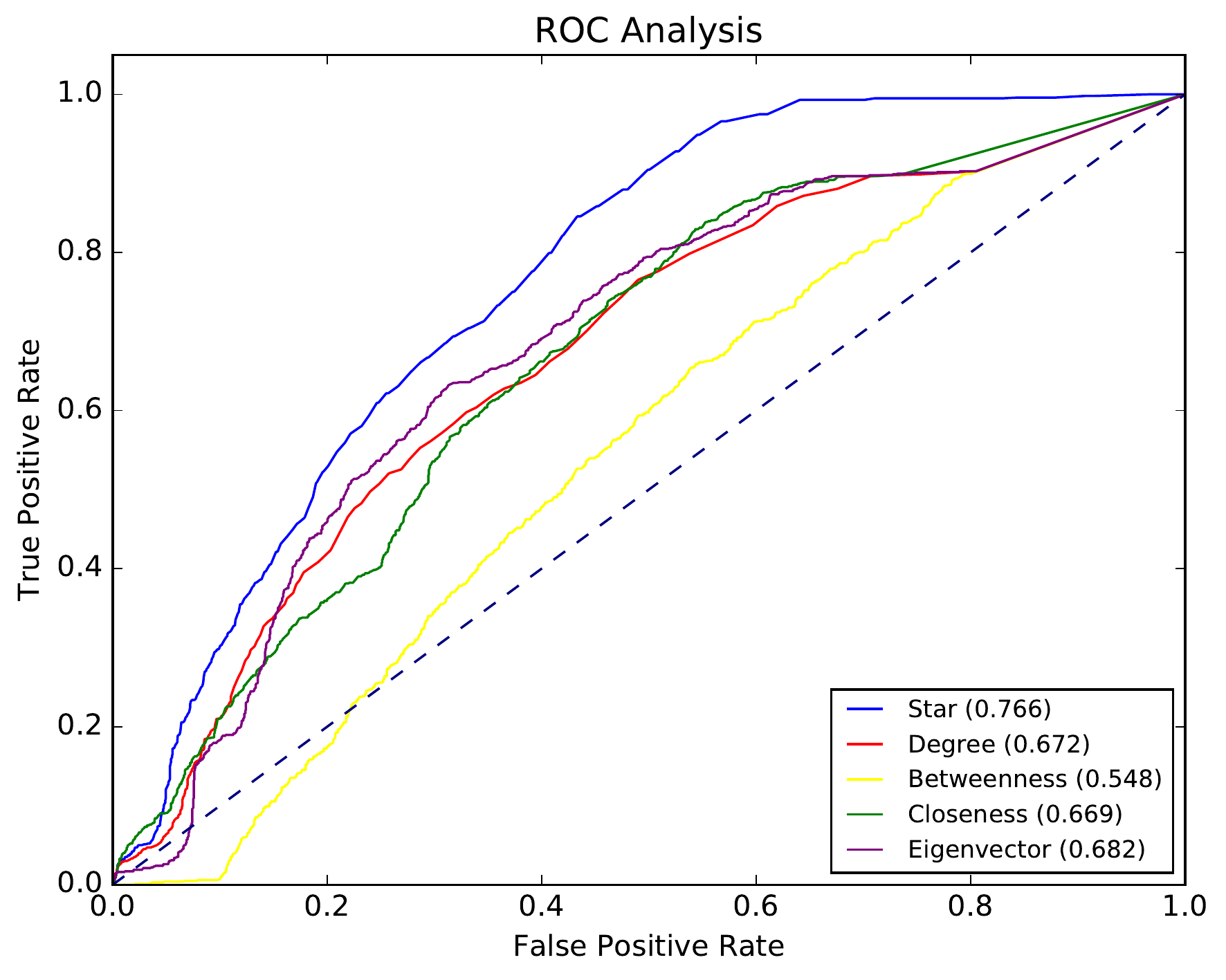}  \includegraphics[width=0.33\textwidth]{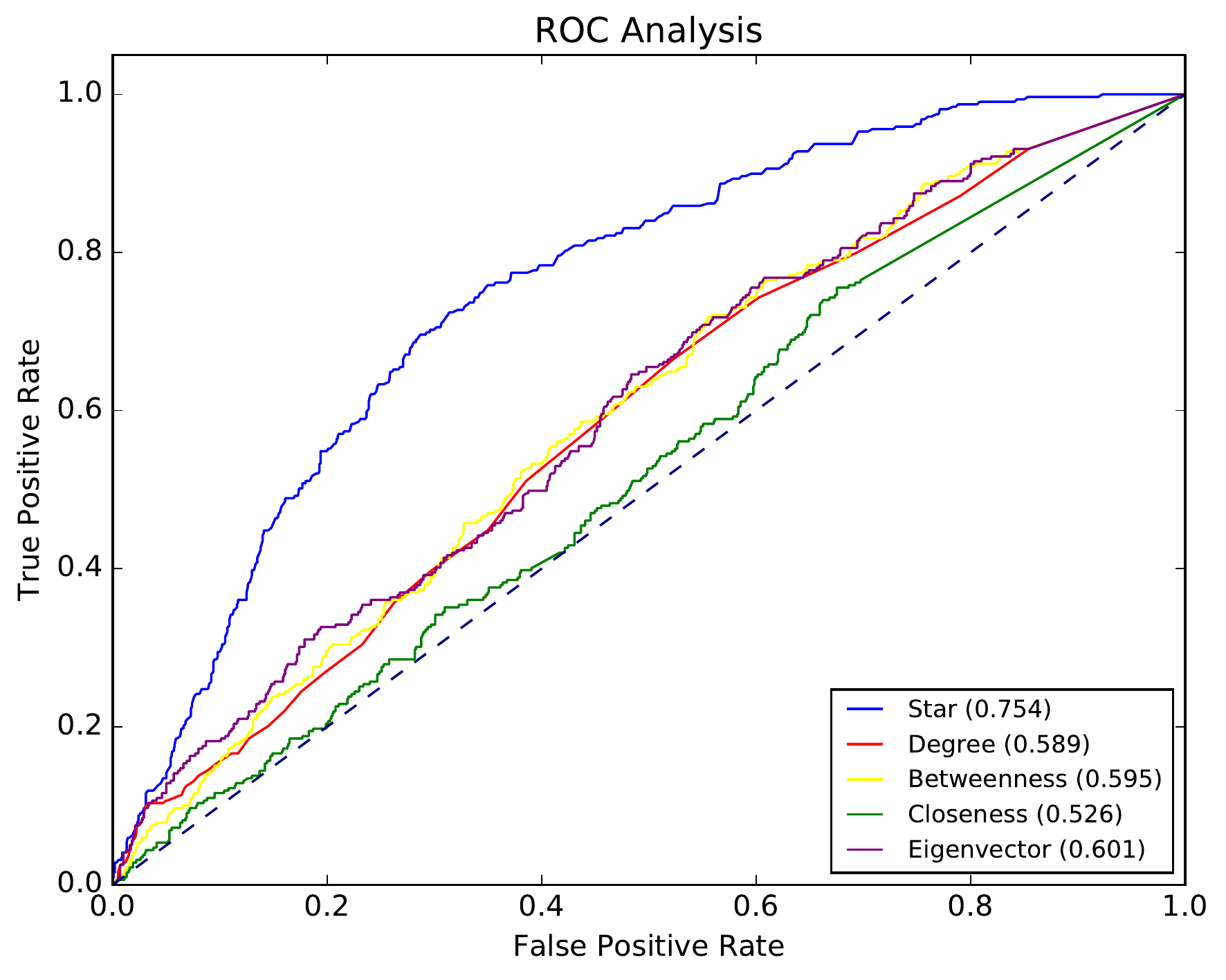}  \includegraphics[width=0.33\textwidth]{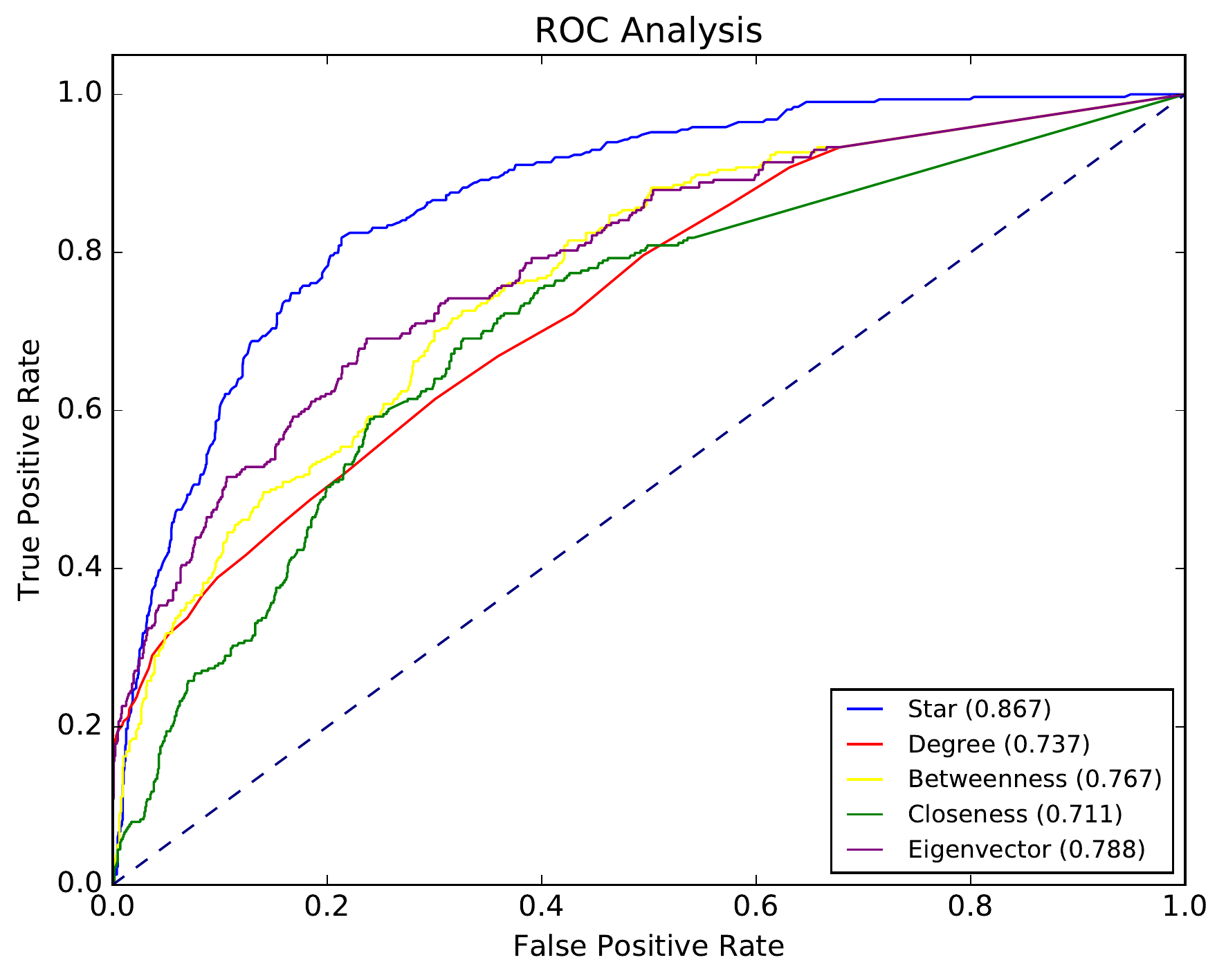}}
	\caption{The Receiver Operating Characteristic curves for each metric for the \emph{Saccharomyces cerevisiae} (yeast), the {\em Helicobacter pylori}, and the {\em Staphylococcus aureus} organisms. All three ROC curves are obtained for a threshold of 60\%. \label{ROC1}}
\end{figure}


\subsection{Sensitivity analysis per metric}
\label{analysis2}

In the second experiment we focus on three of the organisms studied earlier, namely the {\em Saccharomyces cerevisiae}, {\em Helicobacter pylori}, and {\em Staphylococcus aureus} proteomic instances. After generating networks using the threshold as the probability of edge existence, all nodal centrality metrics, along with star centrality, were calculated. Then, the rank of each protein for each metric at every network was calculated in order to find its {\em mean} ranking and its standard deviation. For example, let us assume that only 10 instances were randomly generated and a protein was ranked first in 5 of them, second in 3 of the instances, fourth in 1 of the instances, and fifth in the last one. Such a protein would have an average ranking of $(5\cdot 1+3\cdot 2+1\cdot 4+1\cdot 5)/10=20/10=2$. This enables us to calculate a {\em coefficient of variation} (CV) for each of the proteins in each of the random instances, which can be used to quantify the variability in each of the metrics. 

Finally, to show that star centrality is stable under this random edge existence, we created a box-and-whisker plot (box plot) for each of the three organisms. The results are summarized in Figure \ref{boxplots}. 

\begin{figure}
	\tiny
	\centering
	{\includegraphics[width=0.33\textwidth]{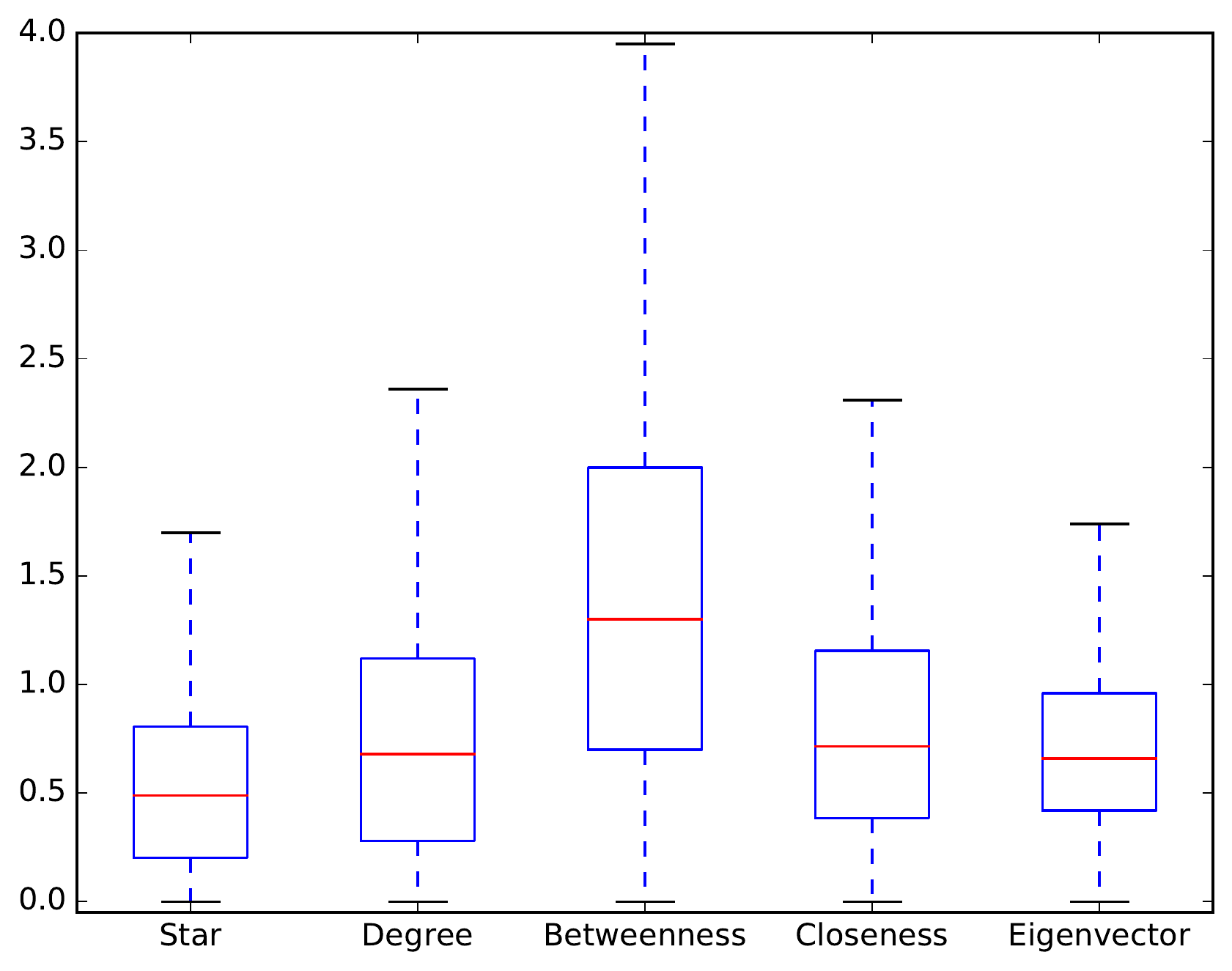}  \includegraphics[width=0.33\textwidth]{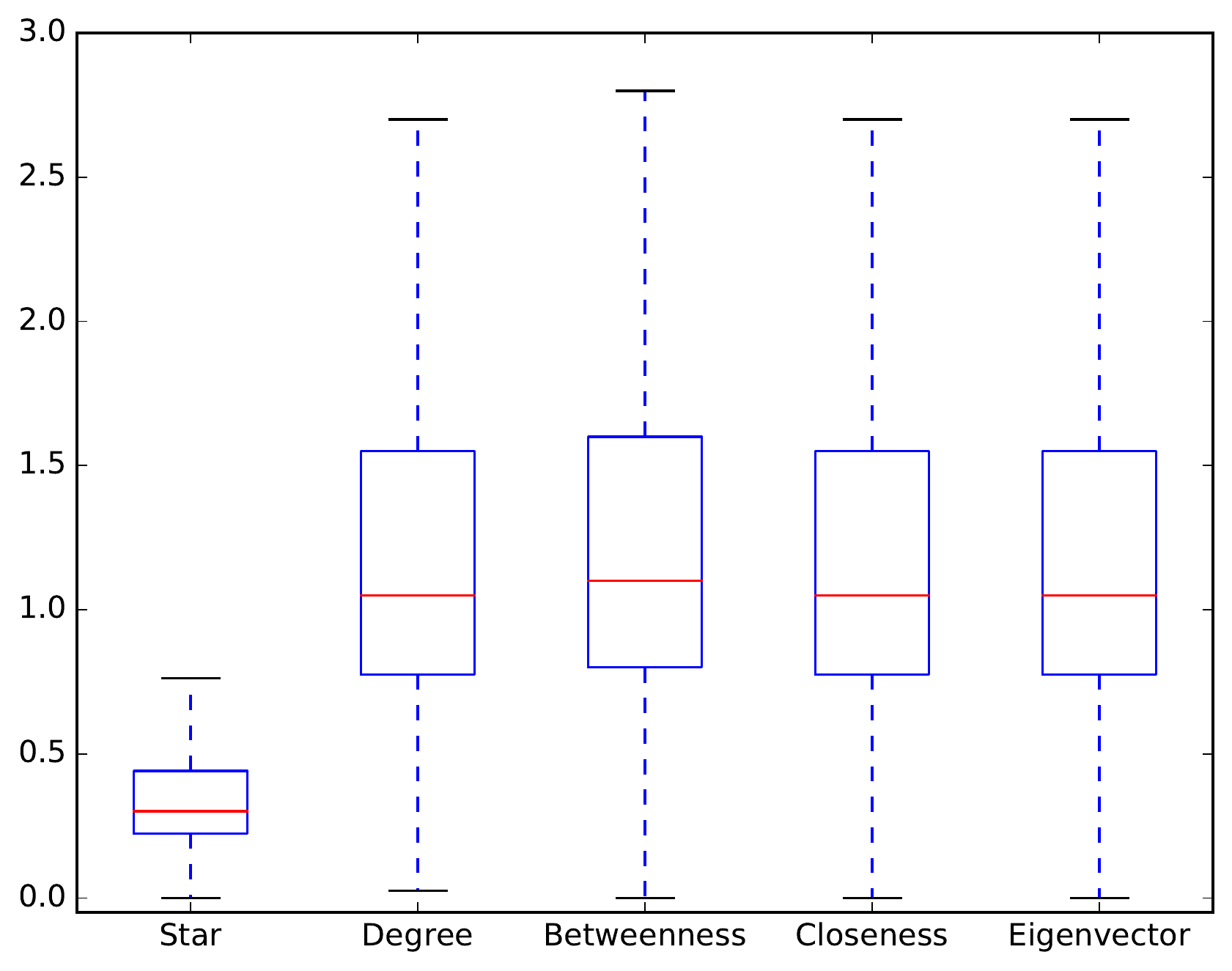}  \includegraphics[width=0.33\textwidth]{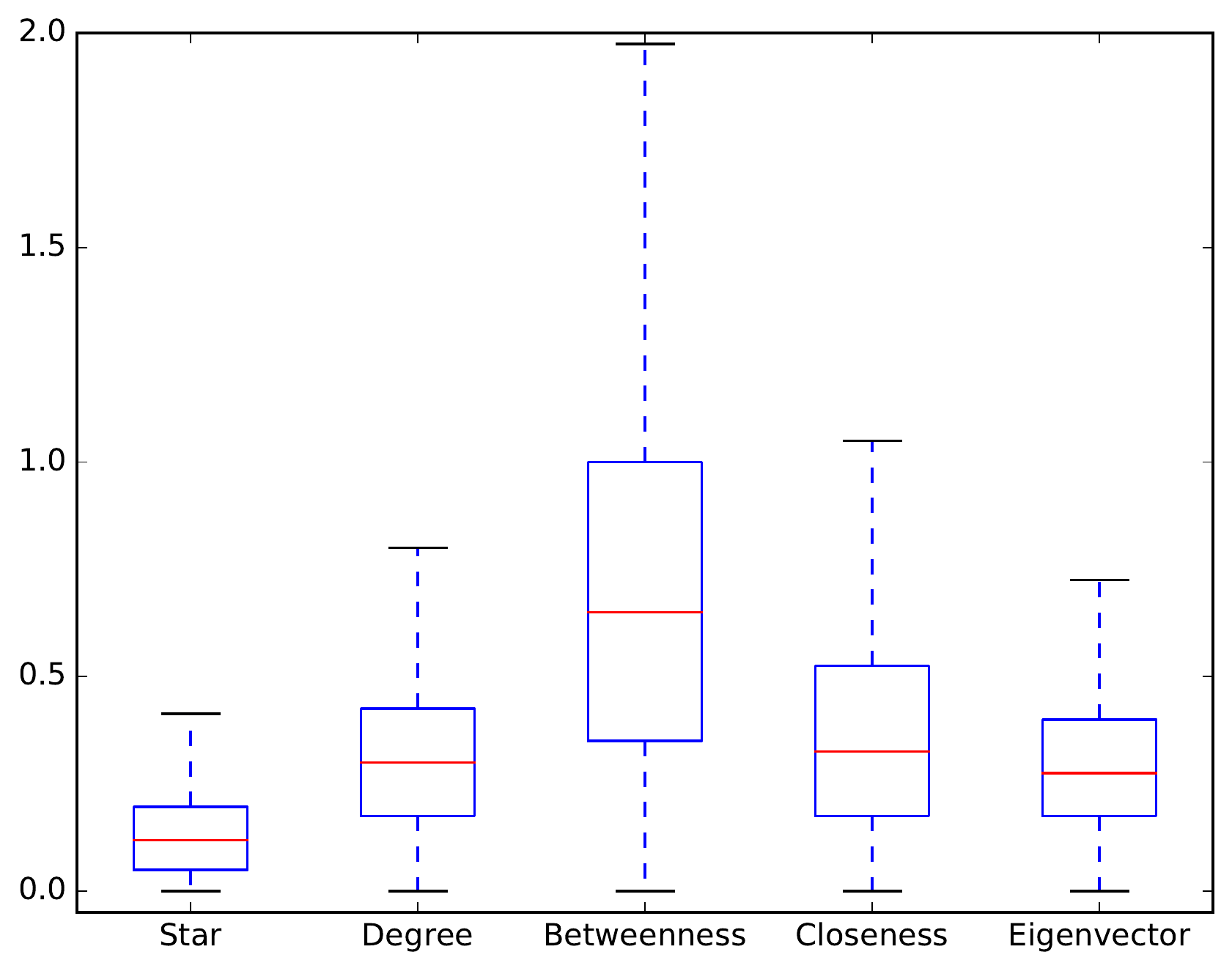}}
	\caption{The box-and-whisker plots for each of the \emph{Saccharomyces cerevisiae} (yeast),  {\em Helicobacter pylori}, and {\em Staphylococcus aureus} organisms, when randomly generating 20 instances of each. We observe that star centrality consistently showcases a lower coefficient of variation, whereas betweenness centrality is the most volatile. In the case of {\em Helicobacter pylori}, all nodal centrality metrics were shown to be similarly unstable in the case of random edge additions/deletions. In the other two organisms, degree centrality was also stable when compared to the rest of the metrics. Let it also be noted that the three plots are not having the same axis limits, and hence are not used to draw any conclusions or comparisons for the proteomes themselves. \label{boxplots}}
\end{figure}

\subsection{Greedy Algorithm Analysis}

\label{greedyAnalysis}

In this subsection, we compare the performance of the two approximation algorithms in practice, using the same PPINs as in subsection \ref{analysis1}. The results are summarized in Tables \ref{tab_approximation} and \ref{tab_times}. We make the following observations. First, Algorithm \ref{RatioGreedy} provides a better solution for every protein in every PPIN when compared to Algorithm \ref{SimpleGreedy}. On average though, as can be seen in Table \ref{tab_approximation}, both algorithms perform similarly well, finding the optimal solution in the majority of proteins. 

\begin{table}
\centering
\tiny
\caption{Approximation ratio analysis for both Algorithms \ref{SimpleGreedy} and \ref{RatioGreedy} for different PPINs. The last columns show the ratio of optimal solutions found. \label{tab_approximation}}
{\begin{tabular}{l|cc|cc|cc}
	 & \multicolumn{2}{c|}{Average Approximation} & \multicolumn{2}{c|}{Minimum Approximation} & \multicolumn{2}{c}{Optimal Found} \\ 
	 \cline{2-7}
	Organism & Simple & Ratio-based & Simple & Ratio-based & Simple & Ratio-based \\
	\hline
	\emph{Saccharomyces cerevisiae} & 0.87 & 0.88 & 0.02 & 0.64 & 0.76 & 0.78 \\
\emph{Helicobacter pylori} &  0.92 & 0.95 & 0.43 & 0.54 & 0.61 & 0.65   \\
\emph{Staphylococcus aureus} &  0.93 & 0.98 & 0.20 & 0.57 & 0.66 & 0.75 \\
\emph{Salmonella enterica CT18} &0.97 & 0.97 &	0.24 & 0.69 & 0.63 & 0.66 \\
\emph{Caenorhabditis elegans}  & 0.95 & 0.99 & 0.051 & 0.60 & 0.78 & 0.79 
\end{tabular}}
\end{table}

More specifically, we note that in all organisms, \emph{Ratio-based Greedy} always found a solution that was at least half as good as the optimal. On the other hand, we note that there are occasions where the \emph{Simple Greedy} fails to get a high quality solution and behaves close to its approximation guarantee. However, we can also observe that both approximation algorithms are able to find solutions that are very close to the optimal. In all organisms the solution obtained by either algorithm was on average as good as 96.3\% of the optimal solution. This means that, even though in some cases the exact optimal is not found, the optimality gap is very small. 

\begin{table}[h]
\centering
\tiny
\caption{Average and maximum computational times (in seconds) observed for the approximation algorithms and the Gurobi solver for different PPINs. \label{tab_times}}
{\begin{tabular}{l|ccc|rrr}
	 & \multicolumn{3}{c|}{Average Time} & \multicolumn{3}{c}{Maximum Time}  \\ 
	 \cline{2-7}
	Organism & Simple & Ratio-based & Solver & Simple & Ratio-based & Solver \\
	\hline
	\emph{Saccharomyces cerevisiae} & 0.05  & 0.11   & 0.15 & 15.33  & 102.25 & 361.22 \\
	\emph{Helicobacter pylori} & 0.03  & 0.05  & 0.05  & 0.50  & 1.88 & 3.71   \\
	\emph{Staphylococcus aureus} & 0.03  & 0.05  & 0.07  & 1.29  &2.45 & 3.02   \\
	\emph{Salmonella enterica CT18}  & 0.04  & 0.16  & 0.56  & 3.62  & 9.14 & 21.93  \\
	\emph{Caenorhabditis elegans}   & 0.09 & 0.34   & 1.13 & 18.15 & 189.32  & 1865.10
\end{tabular}}
\end{table}

As far as our time study, shown in Table \ref{tab_times}, is concerned, the main result is that, as expected, \emph{Simple Greedy} outperforms both the more refined \emph{Ratio-based Greedy} and the Gurobi solver. This performance extends to both the average and the worst-case behavior of the three approaches.

\section{Conclusions}
\label{conclusions}

In this work, we propose a new centrality metric, called \emph{star centrality}, which aims to consider the connections of the ``best" induced star centered at a node $i$. The problem was shown to be $\mathcal{NP}$-hard, however two approximation algorithms that perform efficiently, both as far as execution time and solution quality are concerned, were devised and implemented. The metric was then compared to traditional nodal centrality metrics in real-life protein-protein interaction networks, outperforming them in all instances; often significantly. 

The implications from our work are two-fold. From a biological aspect, this metric provides researchers with a new and improved scoring scheme for ranking proteins and their interactions based on not only the proteins themselves, but also after considering their interacting partners. While our study is focusing on a specific type of clusters (induced stars), understanding how the new score works can prove valuable for developing other, group-based scoring/ranking schemes. Another important aspect of our contribution is that we were able to show that by considering groups of proteins we mitigate known problems with current large-scale proteome databases, improving the quality and robustness of the obtained scores. 

We finally observe that the proposed metric does indeed take care of the three caveats mentioned earlier. First, this extension does not favor proteins that participate in a large number of interactions; instead it merely favors proteins that are located in ``strategic", as far as the network topology is concerned, locations in the proteome. Secondly, if an error exists and an interaction is missing (or present, when it should not be), the effect it has in the metric is alleviated as a set of proteins is considered, instead of singleton proteins. Lastly, proteins with low co-expression that however serve to connect otherwise disconnected protein complexes will have a higher star centrality metric, helping in their identification, contrary to other centrality metrics in use for PPINs.

\section*{Acknowledgements}%
Part of this work was performed when Chrysafis Vogiatzis and Mustafa Can Camur were with the Department of Industrial and Manufacturing Engineering at North Dakota State University. Chrysafis Vogiatzis would like to acknowledge support by the National Science Foundation, grant ND EPSCoR NSF 1355466. 

%
 \section*{APPENDIX}
 \label{appendixA}
 
In this appendix, we present the results for the {\em Helicobacter pylori} and {\em Caenorhabditis elegans} for the 60\% threshold (the experiment in Section \ref{analysis1}), as well as all other thresholds used (70\% and 80\%). All results are presented in the same format: first, the Figures showing the accuracy of prediction in the top and bottom ranked proteins are given, followed by their receiver operation characteristic curves. 
 
We first discuss the results obtained on \emph{Salmonella enterica subspecies CT 18} organism (with a threshold of 60\%) where the star centrality metric performs almost twice as well than any other centrality metric, with a final score of 42.09\%, as can be seen in Figure \ref{SalmonellaTop}. As a comparison, the score that is closest is the one of degree centrality (22.84\%), while closeness, betweenness, and eigenvector centrality are at 15.65\%, 20.63\%, and 17.5\%, respectively. When considering the bottom 500 proteins in Figure \ref{SalmonellaBottom}, once more star centrality with a score of 9.18\% misclassifies less essential proteins than the other centrality metrics at 17.68\%, 16.43\%, 18.78\%, and 19.71\%. Last, Figure \ref{ROC2} (left) reveals that star centrality (with an area under the curve of 0.858), eigenvector centrality (area under the curve equal to 0.788), and betweenness centrality (area under the curve of 0.759) all perform well. 

As mentioned in the introduction, the \emph{C. Elegans} organism is of particular interest as it shares common or homologue proteome to humans. Interestingly, for both organisms, star centrality and closeness centrality perform similarly. First, let us focus on Figures \ref{CElegansTop} and \ref{CElegansBottom}. As can be seen, star centrality barely outperforms closeness centrality (behaving similarly) with a final score of 47.29\% compared to 41.58\%. The other three metrics are far behind with scores of 32.19\%, 30.98\%, and 19.10\% for degree, betweenness, and eigenvector centrality. As far as the bottom 500 ranked proteins are concerned, the corresponding scores are low and we note that a similarly low score is observed for the human proteome too. The scores are 3.38\%, 5.40\%, 4.35\%, 4.89\%, and 5.28\% for the centrality metrics in the order presented in the Figure legends. For the ROC curve, shown in Figure \ref{ROC2} (right), the area under the curve is 0.854 when considering star centrality, with closeness centrality a close second (with an area of 0.784).

\begin{figure}
\centering
	\begin{minipage}{0.45\textwidth}
  	\tiny
	{\includegraphics{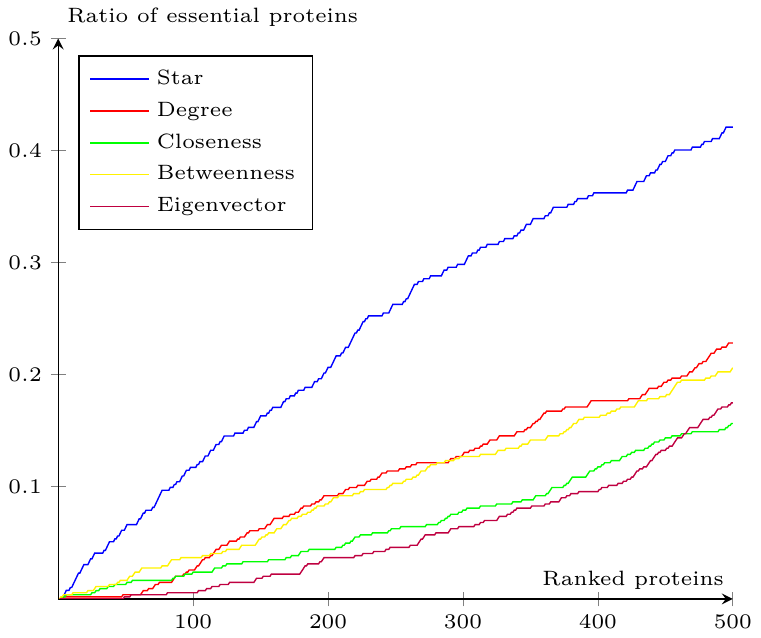}}
	\caption{The ratio of essential proteins detected in the ranked \emph{top k} proteins according to each metric for the \emph{Salmonella enterica CT18} organism when a threshold of 60\% was used. \label{SalmonellaTop}}
  \end{minipage} ~~~%
  \begin{minipage}{0.45\textwidth}
\tiny
	{\includegraphics{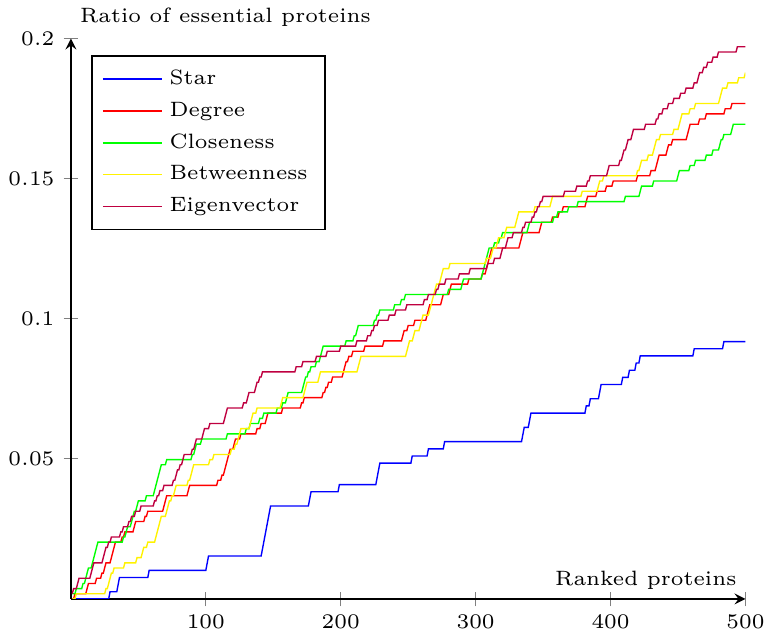}}
	\caption{The ratio of essential proteins detected in the ranked \emph{bottom k} proteins according to each metric for the \emph{Salmonella enterica CT18} organism when a threshold of 60\% was used. \label{SalmonellaBottom}}
  \end{minipage}
\end{figure}

\begin{figure}
\centering
	\begin{minipage}{0.45\textwidth}
  	\tiny
	{\includegraphics{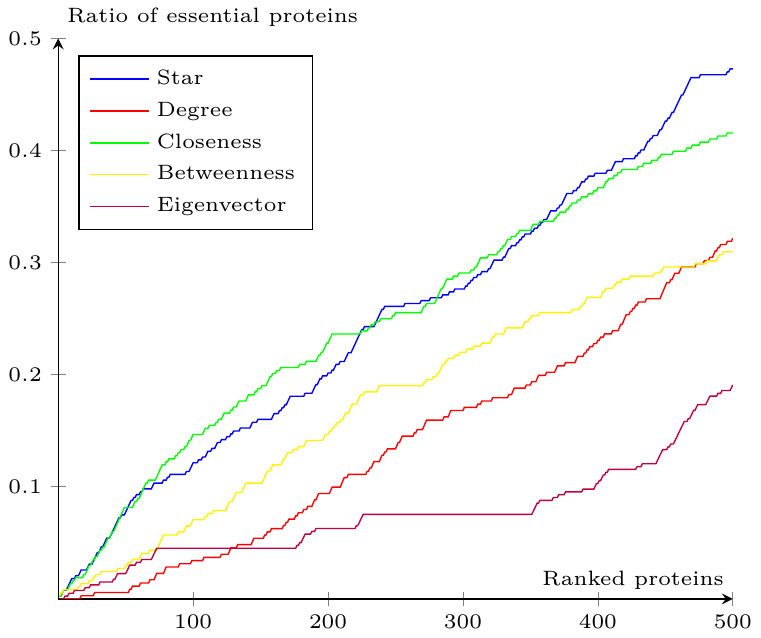}}
	\caption{The ratio of essential proteins detected in the ranked \emph{top k} proteins according to each metric for the \emph{C. Elegans} organism when a threshold of 60\% was used. \label{CElegansTop}}
  \end{minipage} ~~~%
  \begin{minipage}{0.45\textwidth}
\tiny
	{\includegraphics{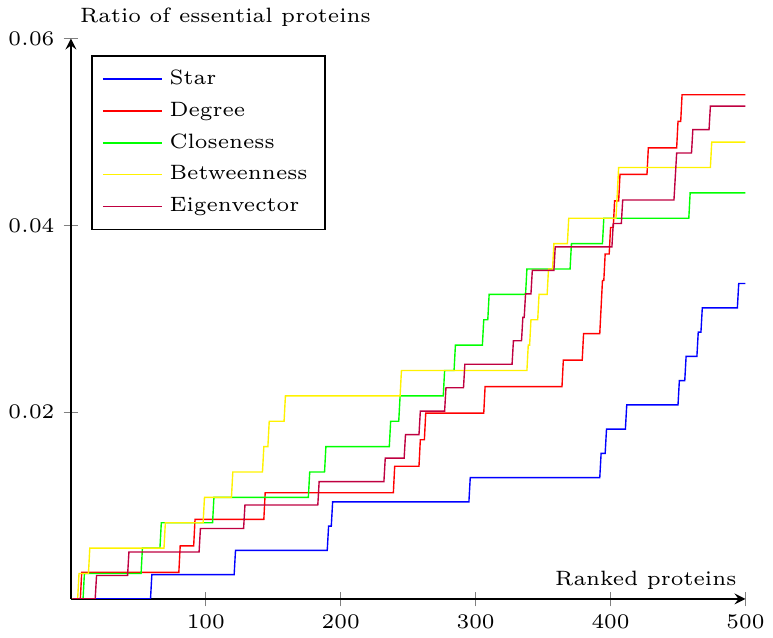}}
	\caption{The ratio of essential proteins detected in the ranked \emph{bottom k} proteins according to each metric for the \emph{C. Elegans} organism when a threshold of 60\% was used. \label{CElegansBottom}}
  \end{minipage}
\end{figure}

\begin{figure}
	\tiny
	\centering
	{\includegraphics[width=0.33\textwidth]{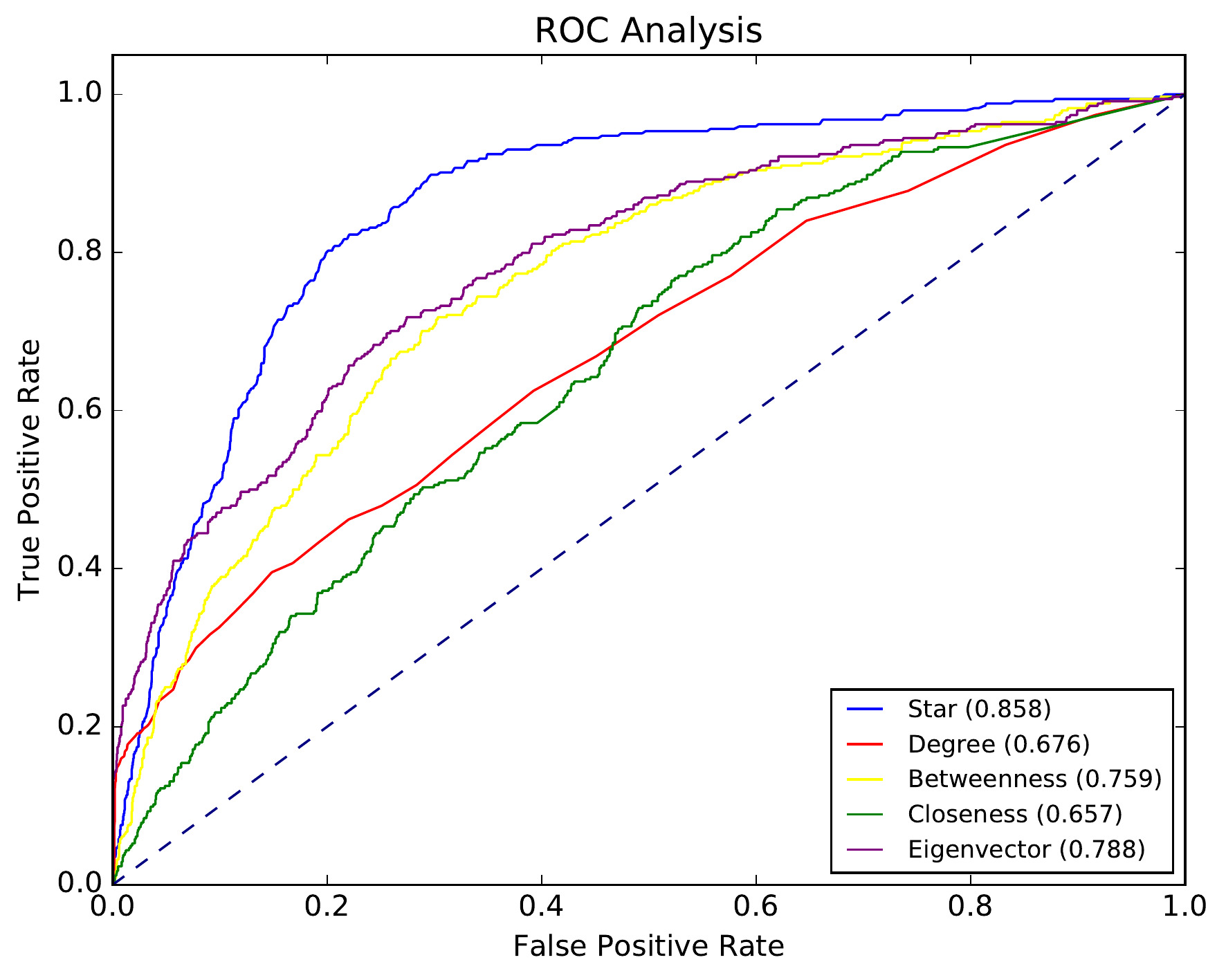}  \includegraphics[width=0.33\textwidth]{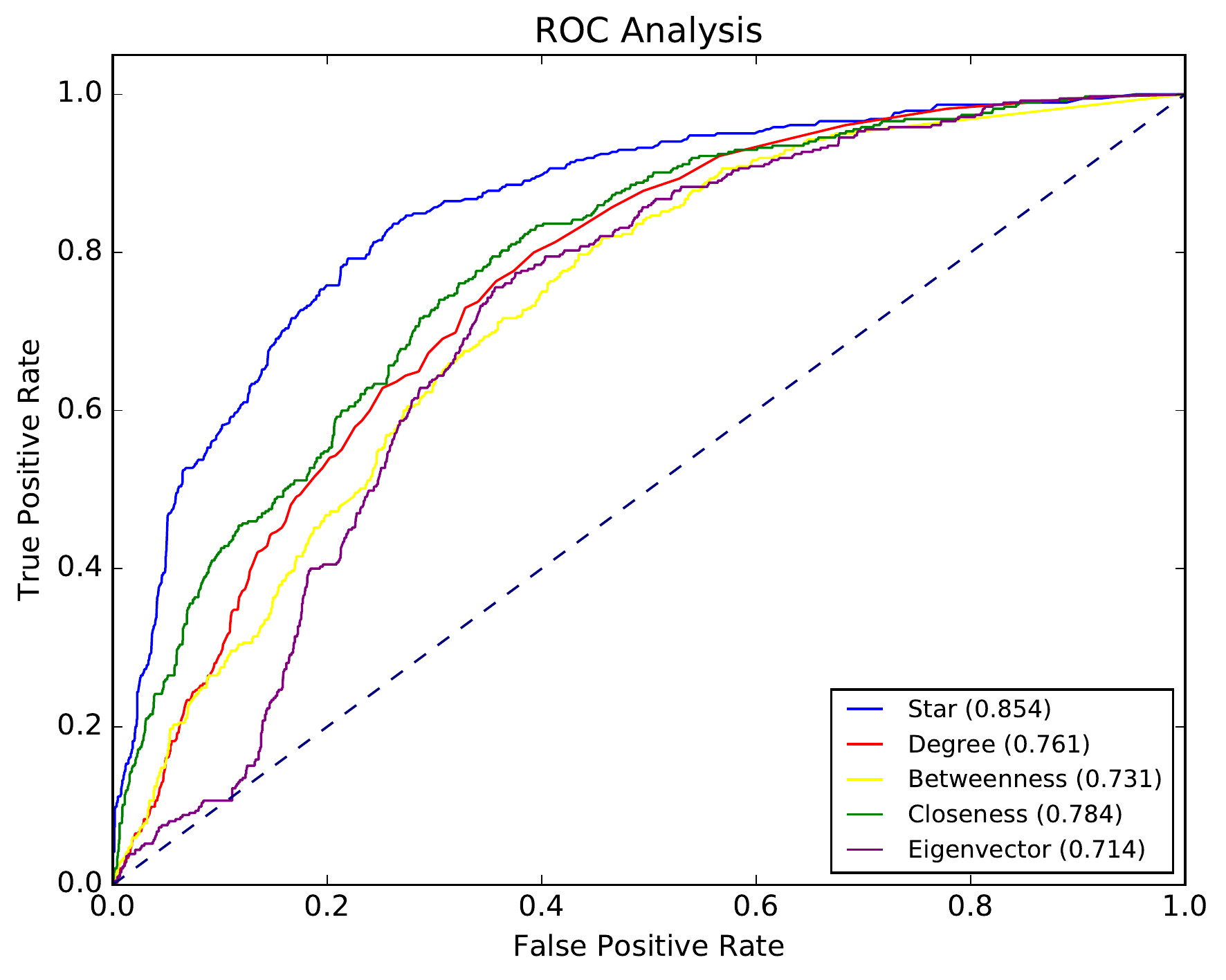}}
	\caption{The Receiver Operating Characteristic curves for each metric for the \emph{Salmonella enterica CT18} (yeast) and the {\em Caenorhabditis elegans} organisms. Both ROC curves are obtained for a threshold of 60\%. \label{ROC2}}
\end{figure}

The remainder of the Appendix presents all other experimental configurations between all five organisms and the remaining thresholds. The results are shown in Figures \ref{SCerevisiaeTop_700}-\ref{CElegansBottom_800} (for the top and bottom $k$ proteins analysis), followed by the receiver operation characteristic curves in Figures \ref{ROC7001}-\ref{ROC8002}.

 \begin{figure}
	\begin{minipage}{0.45\textwidth}
  	\tiny
	{\includegraphics{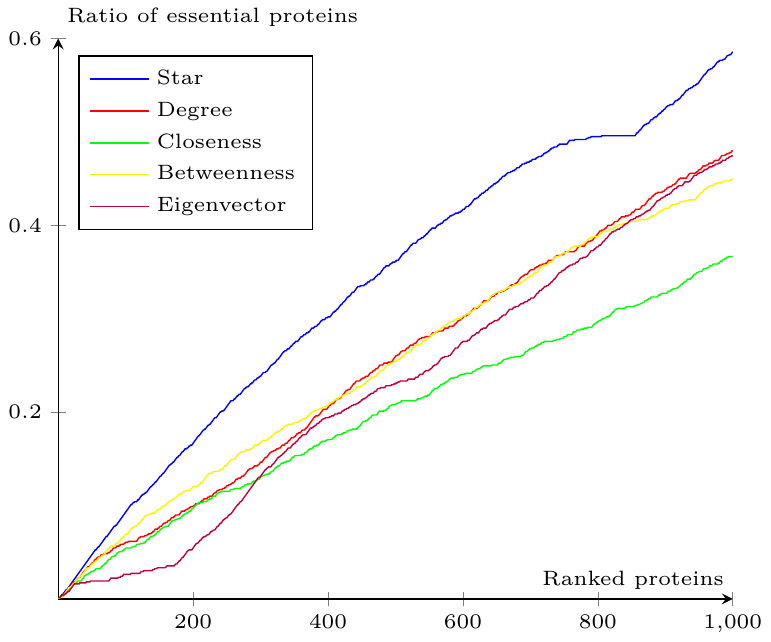}}
	\caption{The ratio of essential proteins detected in the ranked \emph{top k} proteins according to each metric for the \emph{Saccharomyces cerevisiae} organism (yeast) when a threshold of 70\% was used. \label{SCerevisiaeTop_700}}
  \end{minipage} ~~~%
  \begin{minipage}{0.45\textwidth}
\tiny
	{\includegraphics{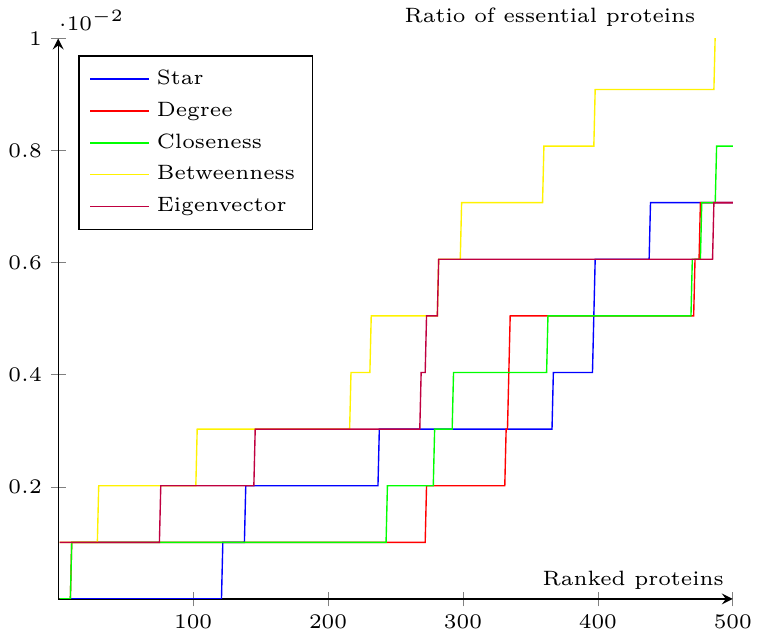}}
	\caption{The ratio of essential proteins detected in the ranked \emph{bottom k} proteins according to each metric for the \emph{Saccharomyces cerevisiae} organism (yeast) when a threshold of 70\% was used. \label{SCerevisiaeBottom_700}}
  \end{minipage}
\end{figure}

\begin{figure}
	\begin{minipage}{0.45\textwidth}
  	\tiny
	{\includegraphics{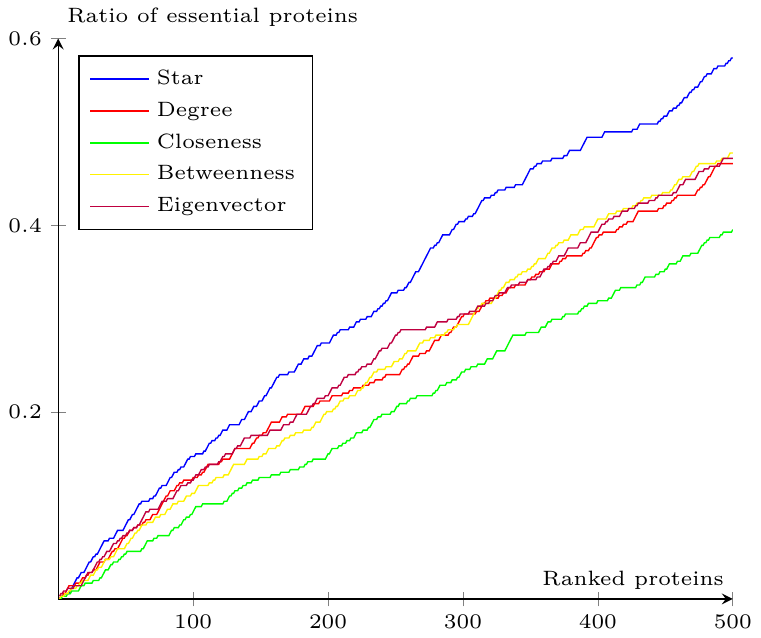}}
	\caption{The ratio of essential proteins detected in the ranked \emph{top k} proteins according to each metric for the \emph{Helicobacter pylori} organism when a threshold of 70\% was used. \label{HPylorisTop_700}}
  \end{minipage} ~~~%
  \begin{minipage}{0.45\textwidth}
\tiny
	{\includegraphics{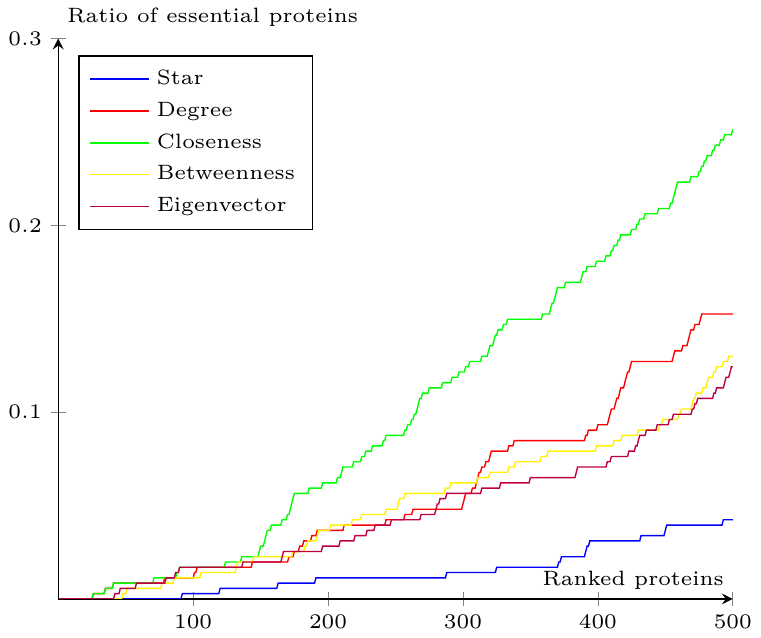}}
	\caption{The ratio of essential proteins detected in the ranked \emph{bottom k} proteins according to each metric for the \emph{Helicobacter pylori} organism when a threshold of 70\% was used. \label{HPylorisBottom_700}}
  \end{minipage}
\end{figure}

\begin{figure}
	\begin{minipage}{0.45\textwidth}
  	\tiny
	{\includegraphics{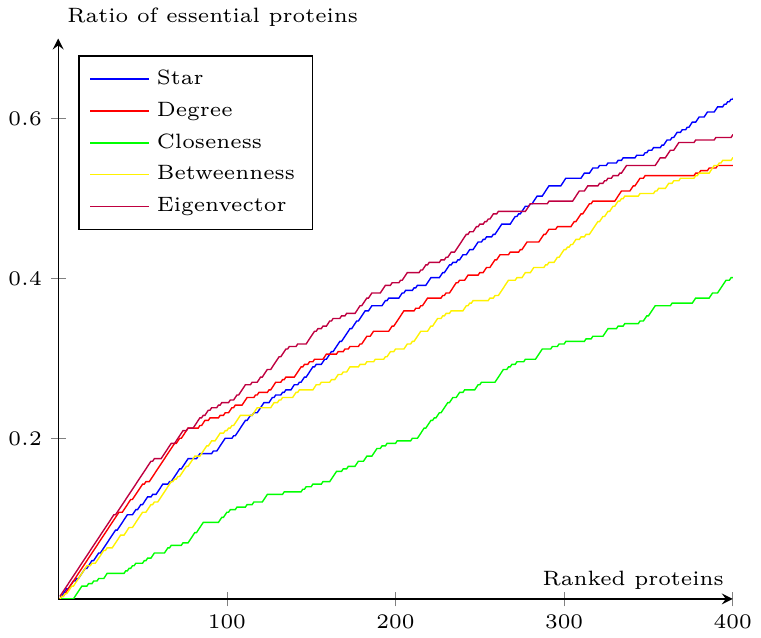}}
	\caption{The ratio of essential proteins detected in the ranked \emph{top k} proteins according to each metric for the \emph{Staphylococcus aureus} organism when a threshold of 70\% was used. \label{StaphTop_700}}
  \end{minipage} ~~~%
  \begin{minipage}{0.45\textwidth}
\tiny
	{\includegraphics{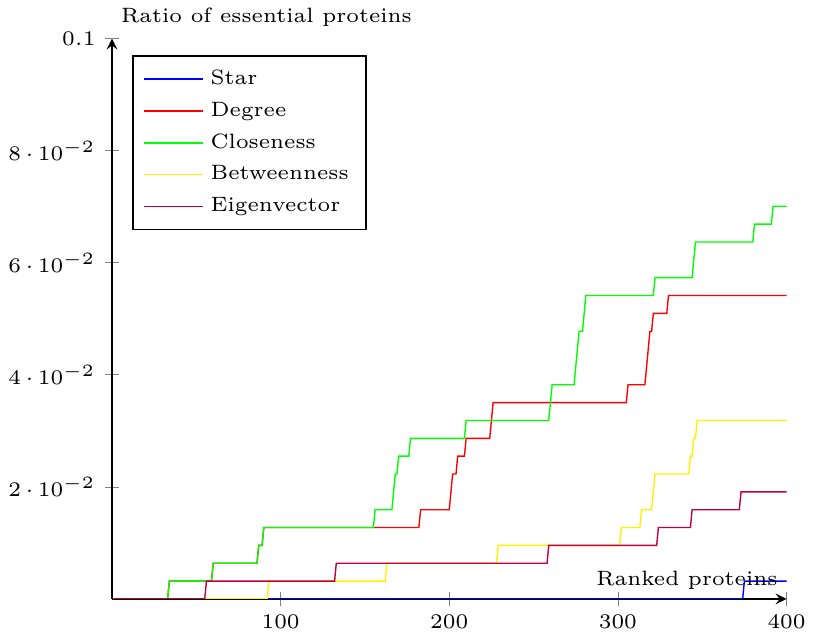}}
	\caption{The ratio of essential proteins detected in the ranked \emph{bottom k} proteins according to each metric for the \emph{Staphylococcus aureus} organism when a threshold of 70\% was used. \label{StaphBottom_700}}
  \end{minipage}
\end{figure}

\begin{figure}
	\begin{minipage}{0.45\textwidth}
  	\tiny
	{\includegraphics{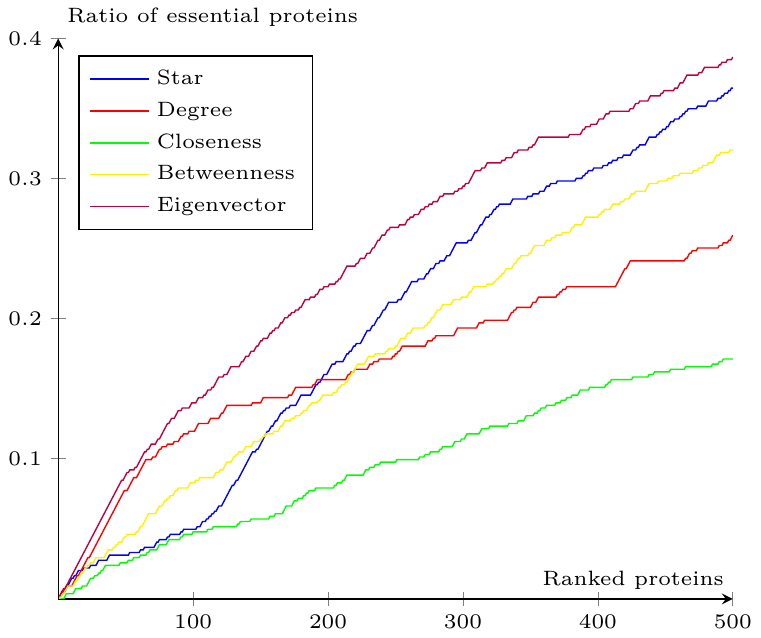}}
	\caption{The ratio of essential proteins detected in the ranked \emph{top k} proteins according to each metric for the \emph{Salmonella enterica CT17} organism when a threshold of 70\% was used. \label{SalmonellaTop_700}}
  \end{minipage} ~~~%
  \begin{minipage}{0.45\textwidth}
\tiny
	{\includegraphics{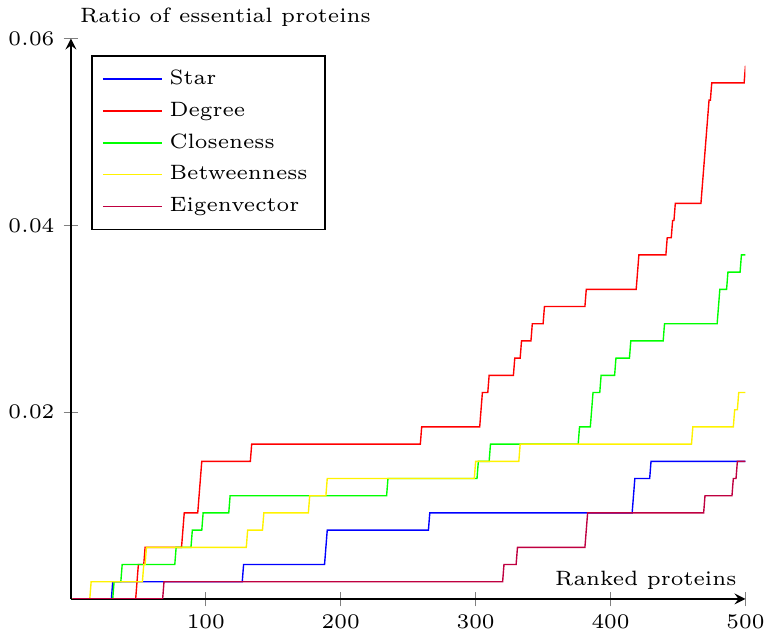}}
	\caption{The ratio of essential proteins detected in the ranked \emph{bottom k} proteins according to each metric for the \emph{Salmonella enterica CT17} organism when a threshold of 70\% was used. \label{SalmonellaBottom_700}}
  \end{minipage}
\end{figure}

\begin{figure}
	\begin{minipage}{0.45\textwidth}
  	\tiny
	{\includegraphics{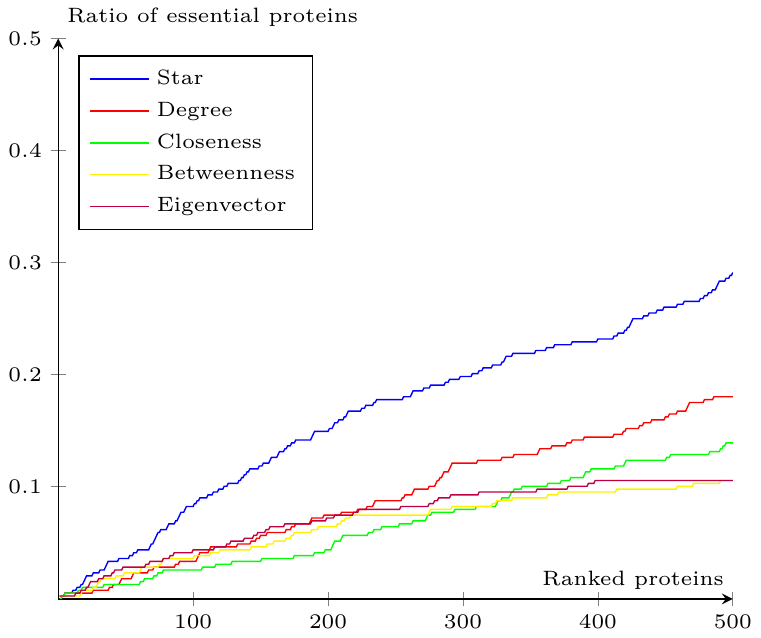}}
	\caption{The ratio of essential proteins detected in the ranked \emph{top k} proteins according to each metric for the \emph{C. Elegans} organism when a threshold of 70\% was used. \label{CElegansTop_700}}
  \end{minipage} ~~~%
  \begin{minipage}{0.45\textwidth}
\tiny
	{\includegraphics{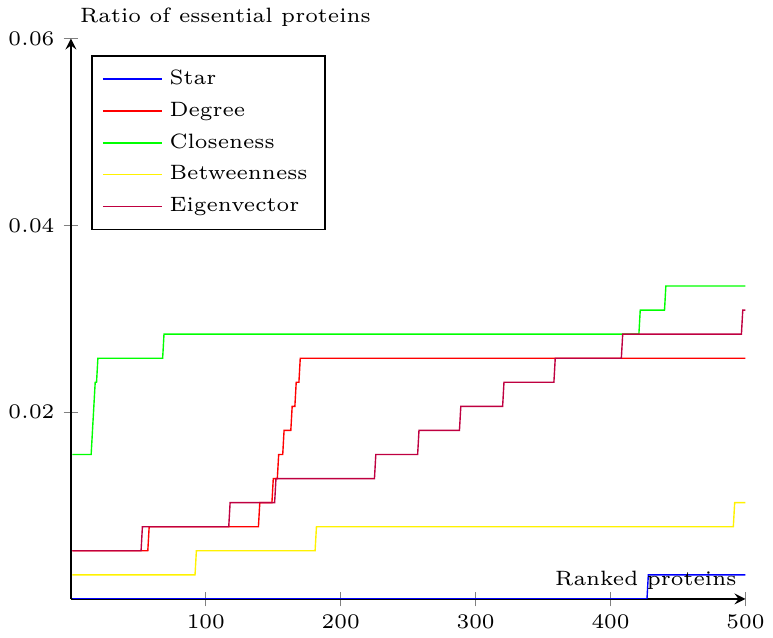}}
	\caption{The ratio of essential proteins detected in the ranked \emph{bottom k} proteins according to each metric for the \emph{C. Elegans} organism when a threshold of 70\% was used. \label{CElegansBottom_700}}
	{}
  \end{minipage}
\end{figure}


 \begin{figure}
	\begin{minipage}{0.45\textwidth}
  	\tiny
	{\includegraphics{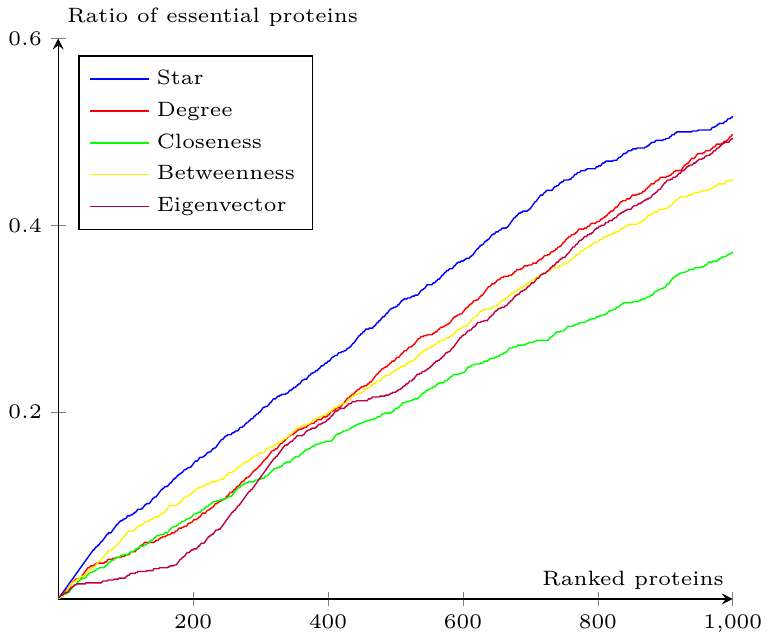}}
	\caption{The ratio of essential proteins detected in the ranked \emph{top k} proteins according to each metric for the \emph{Saccharomyces cerevisiae} organism (yeast) when a threshold of 80\% was used. \label{SCerevisiaeTop_800}}
  \end{minipage} ~~~%
  \begin{minipage}{0.45\textwidth}
\tiny
	{\includegraphics{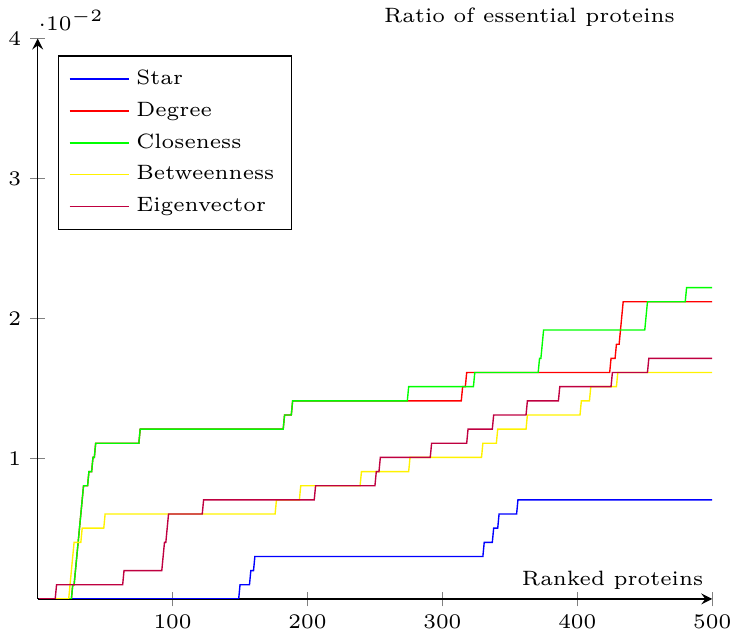}}
	\caption{The ratio of essential proteins detected in the ranked \emph{bottom k} proteins according to each metric for the \emph{Saccharomyces cerevisiae} organism (yeast) when a threshold of 80\% was used. \label{SCerevisiaeBottom_800}}
  \end{minipage}
\end{figure}

\begin{figure}
	\begin{minipage}{0.45\textwidth}
  	\tiny
	{\includegraphics{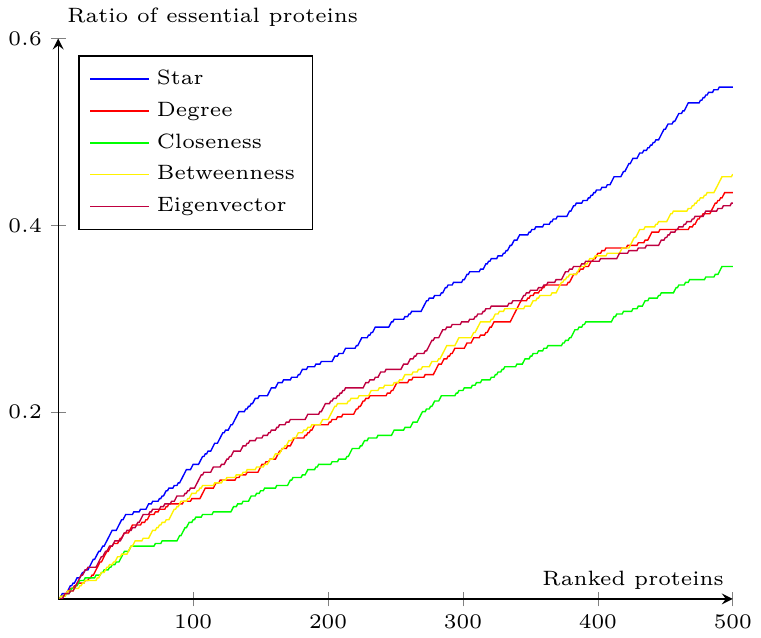}}
	\caption{The ratio of essential proteins detected in the ranked \emph{top k} proteins according to each metric for the \emph{Helicobacter pylori} organism when a threshold of 80\% was used. \label{HPylorisTop_800}}
  \end{minipage} ~~~%
  \begin{minipage}{0.45\textwidth}
\tiny
	{\includegraphics{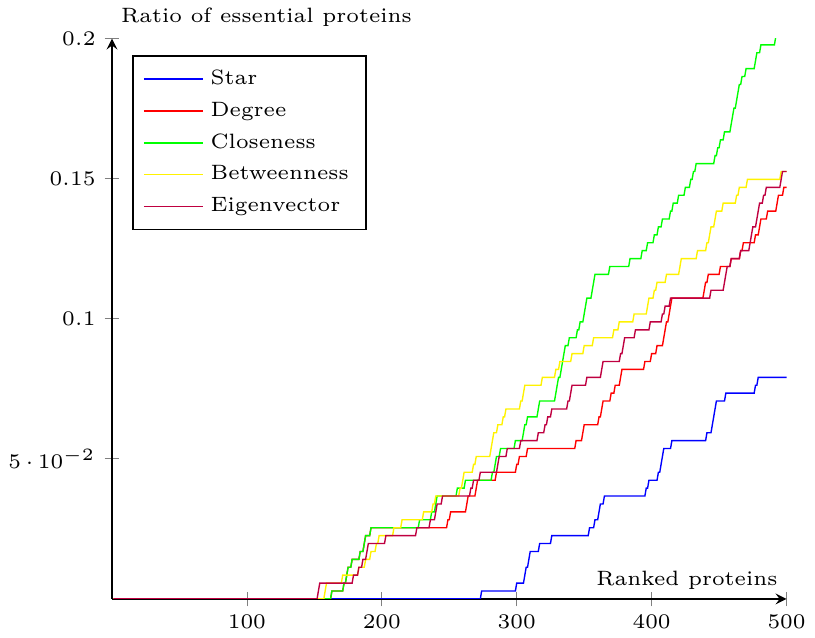}}
	\caption{The ratio of essential proteins detected in the ranked \emph{bottom k} proteins according to each metric for the \emph{Helicobacter pylori} organism when a threshold of 80\% was used. \label{HPylorisBottom_800}}
  \end{minipage}
\end{figure}

\begin{figure}
	\begin{minipage}{0.45\textwidth}
  	\tiny
	{\includegraphics{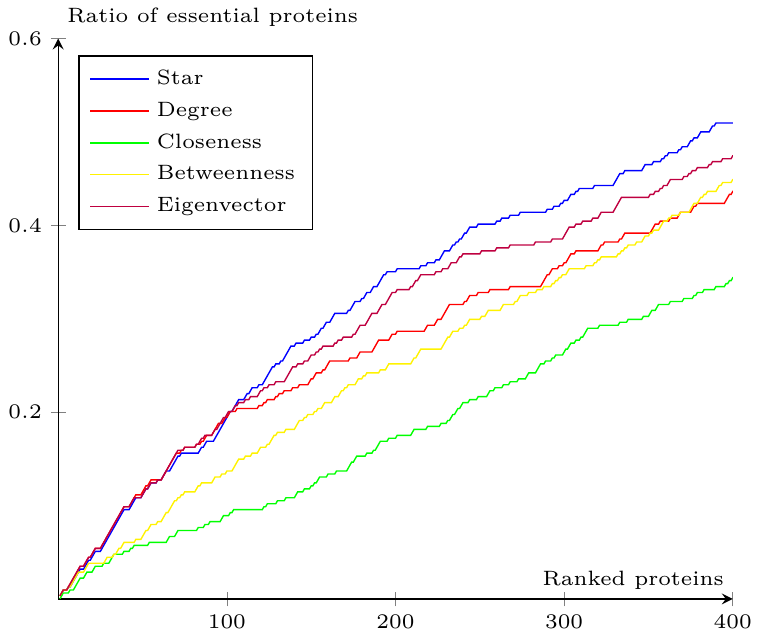}}
	\caption{The ratio of essential proteins detected in the ranked \emph{top k} proteins according to each metric for the \emph{Staphylococcus aureus} organism when a threshold of 80\% was used. \label{StaphTop_800}}
  \end{minipage} ~~~%
  \begin{minipage}{0.45\textwidth}
\tiny
	{\includegraphics{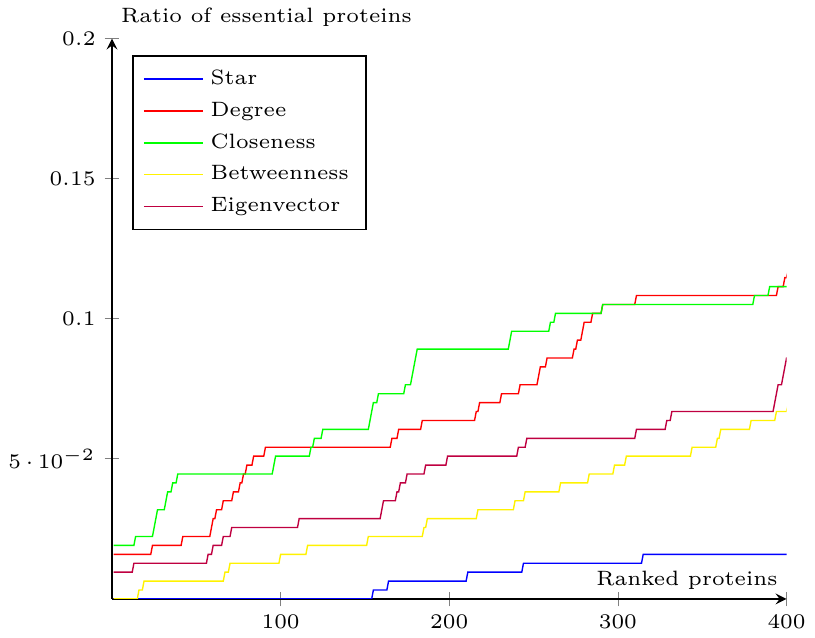}}
	\caption{The ratio of essential proteins detected in the ranked \emph{bottom k} proteins according to each metric for the \emph{Staphylococcus aureus} organism when a threshold of 80\% was used. \label{StaphBottom_800}}
  \end{minipage}
\end{figure}

\begin{figure}
	\begin{minipage}{0.45\textwidth}
  	\tiny
	{\includegraphics{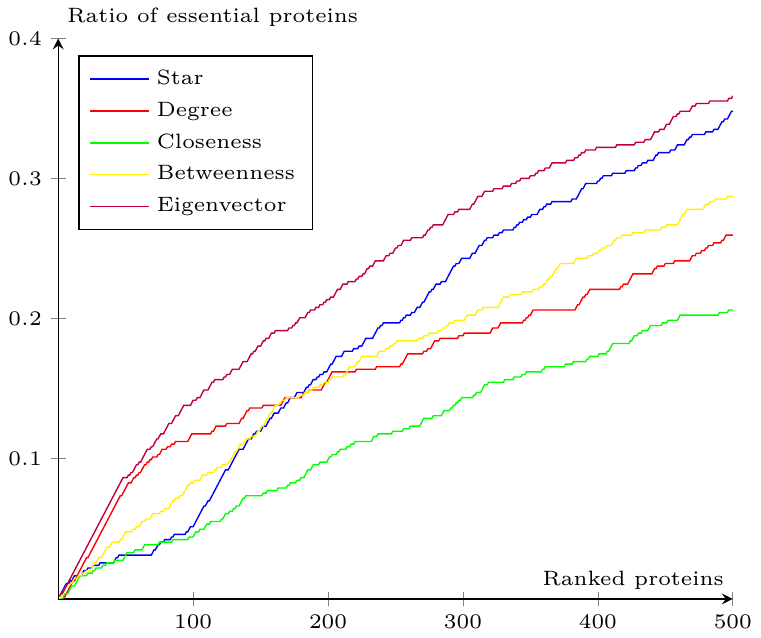}}
	\caption{The ratio of essential proteins detected in the ranked \emph{top k} proteins according to each metric for the \emph{Salmonella enterica CT18} organism when a threshold of 80\% was used. \label{SalmonellaTop_800}}
  \end{minipage} ~~~%
  \begin{minipage}{0.45\textwidth}
\tiny
	{\includegraphics{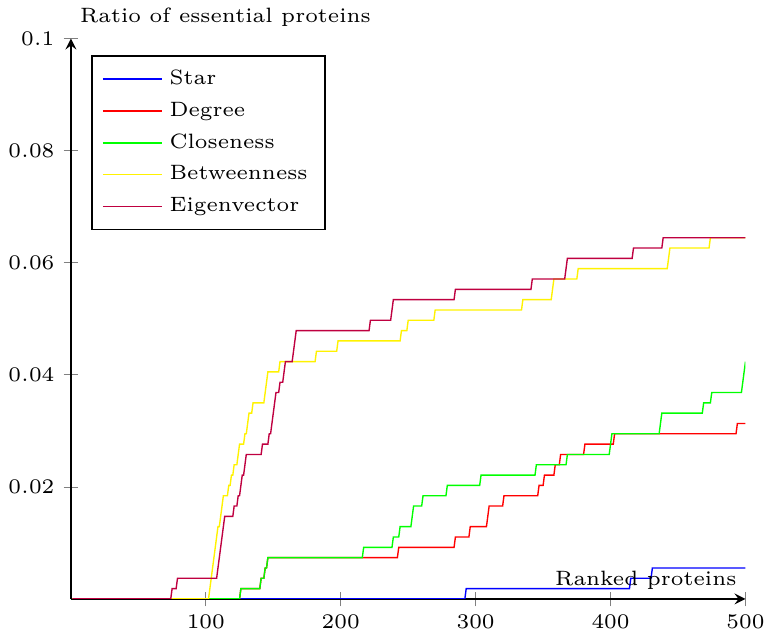}}
	\caption{The ratio of essential proteins detected in the ranked \emph{bottom k} proteins according to each metric for the \emph{Salmonella enterica CT18} organism when a threshold of 80\% was used. \label{SalmonellaBottom_800}}
  \end{minipage}
\end{figure}

\begin{figure}
	\begin{minipage}{0.45\textwidth}
  	\tiny
	{\includegraphics{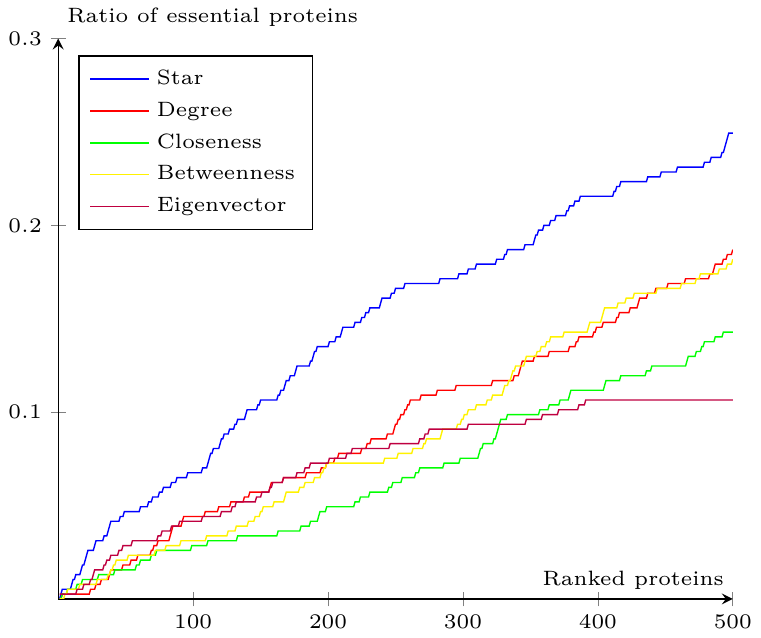}}
	\caption{The ratio of essential proteins detected in the ranked \emph{top k} proteins according to each metric for the \emph{C. Elegans} organism when a threshold of 80\% was used. \label{CElegansTop_800}}
  \end{minipage} ~~~%
  \begin{minipage}{0.45\textwidth}
\tiny
	{\includegraphics{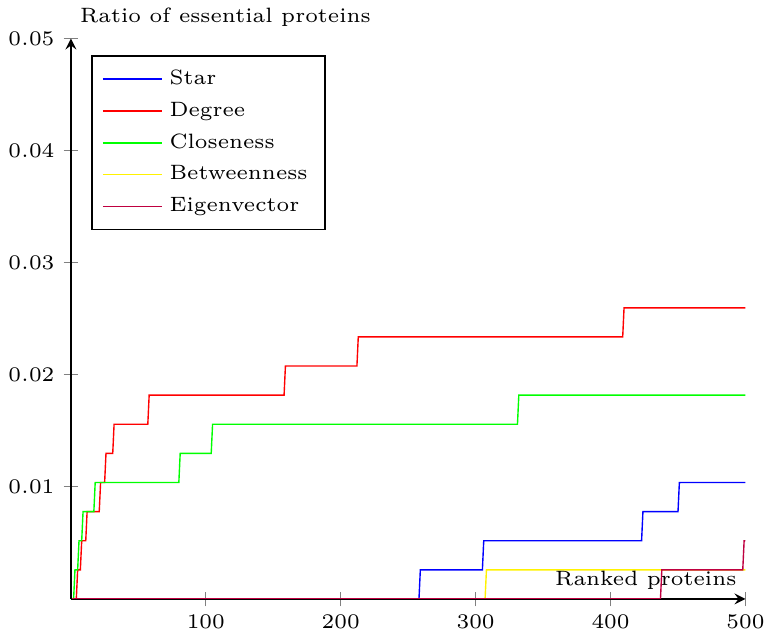}}
	\caption{The ratio of essential proteins detected in the ranked \emph{bottom k} proteins according to each metric for the \emph{C. Elegans} organism when a threshold of 80\% was used. \label{CElegansBottom_800}}
  \end{minipage}
\end{figure}



\begin{figure}
	\tiny
	\centering
	{\includegraphics[width=0.33\textwidth]{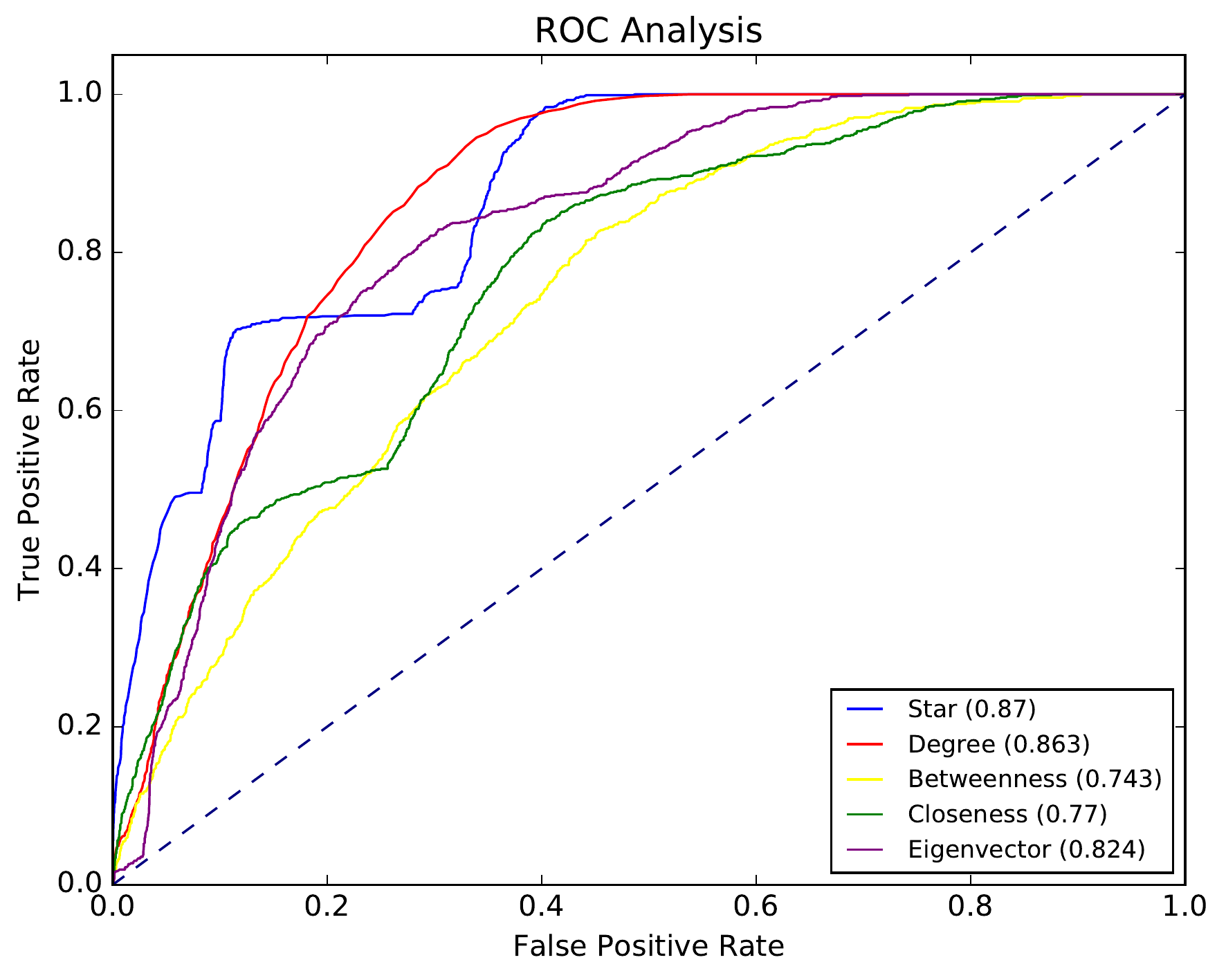}  \includegraphics[width=0.33\textwidth]{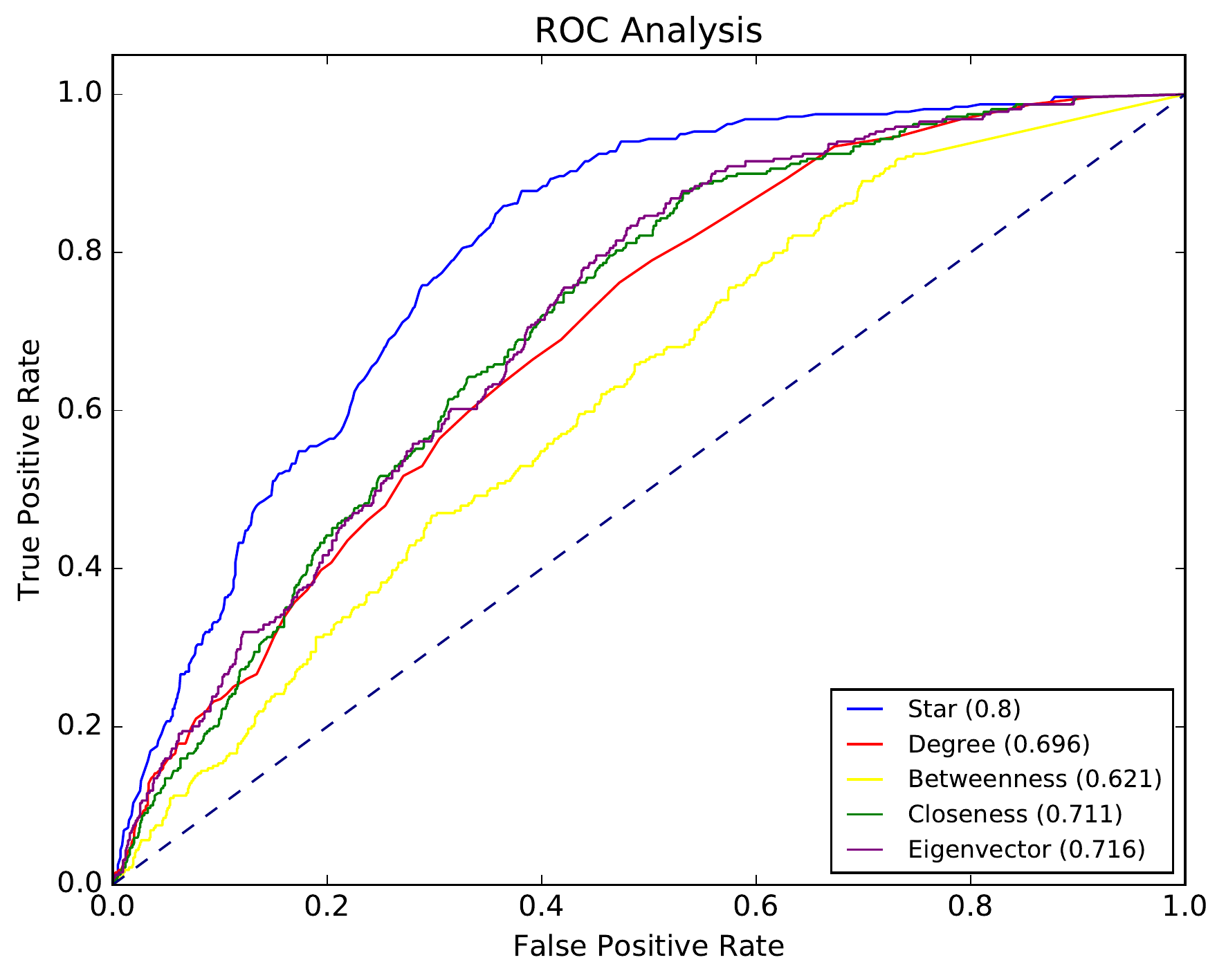}  \includegraphics[width=0.33\textwidth]{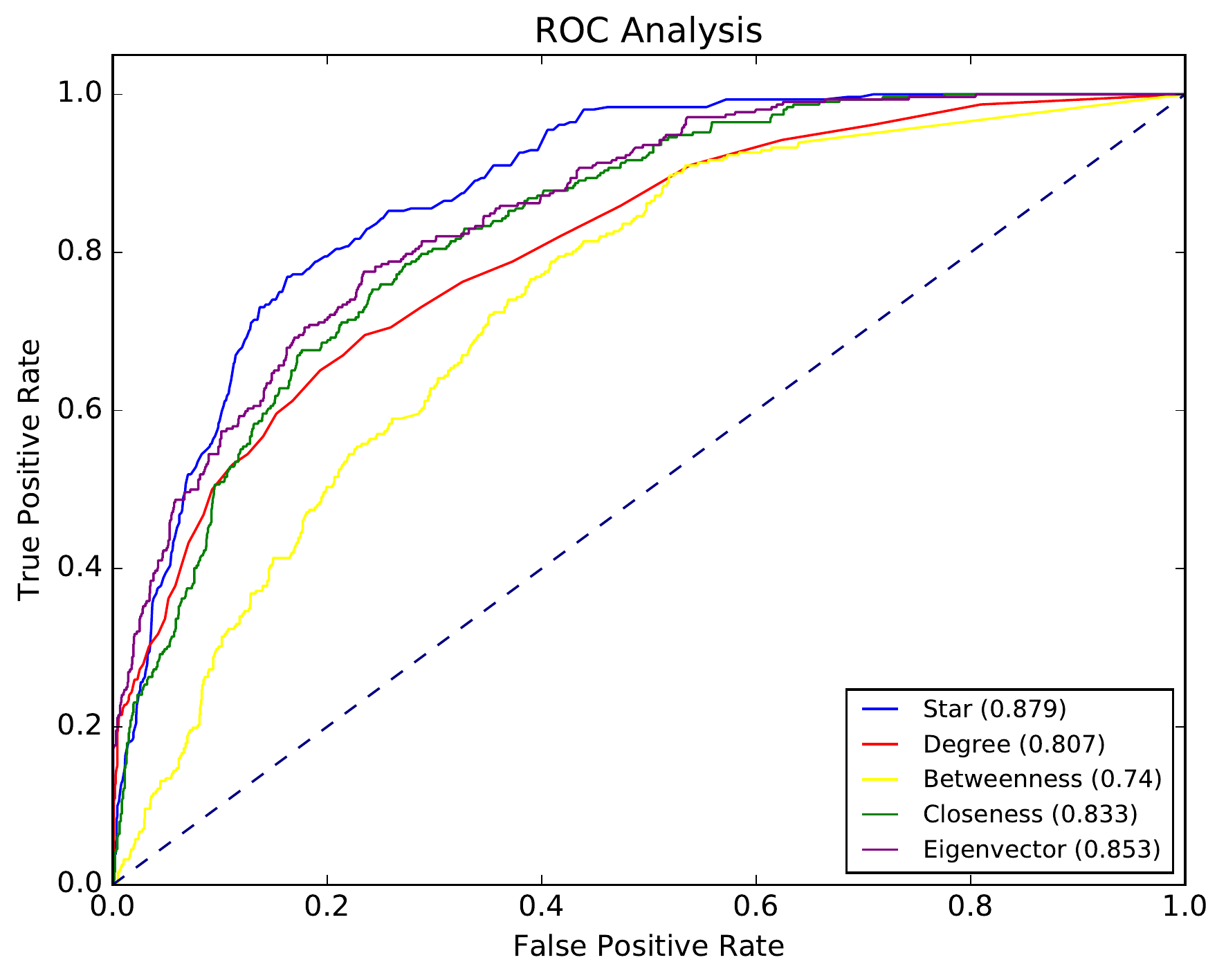}}
	\caption{The Receiver Operating Characteristic curves for each metric for the \emph{Saccharomyces cerevisiae} (yeast), the {\em Helicobacter pylori}, and the {\em Staphylococcus aureus} organisms with a threshold of 70\%. \label{ROC7001}}
\end{figure}

\begin{figure}
	\tiny
	\centering
	{\includegraphics[width=0.33\textwidth]{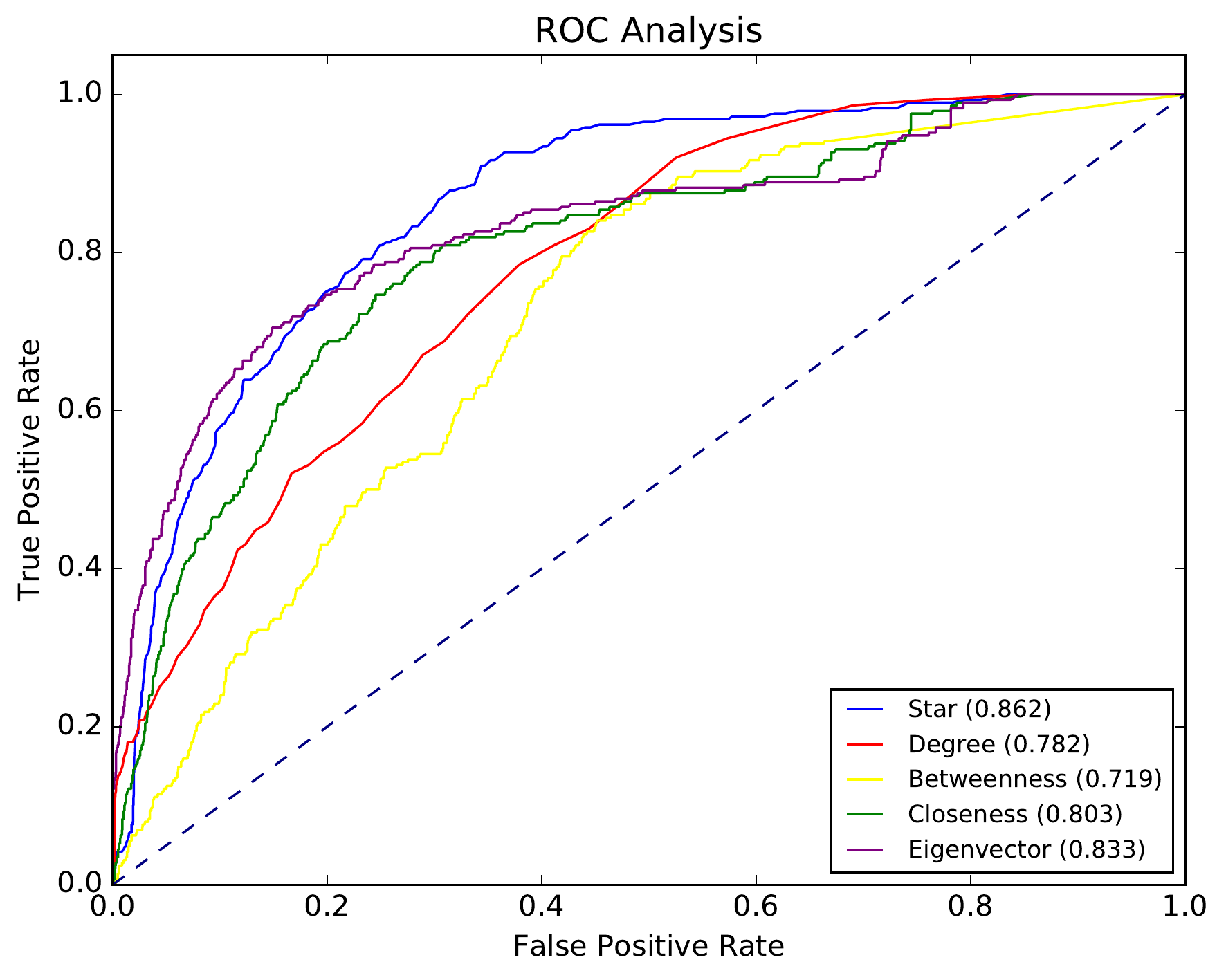}  \includegraphics[width=0.33\textwidth]{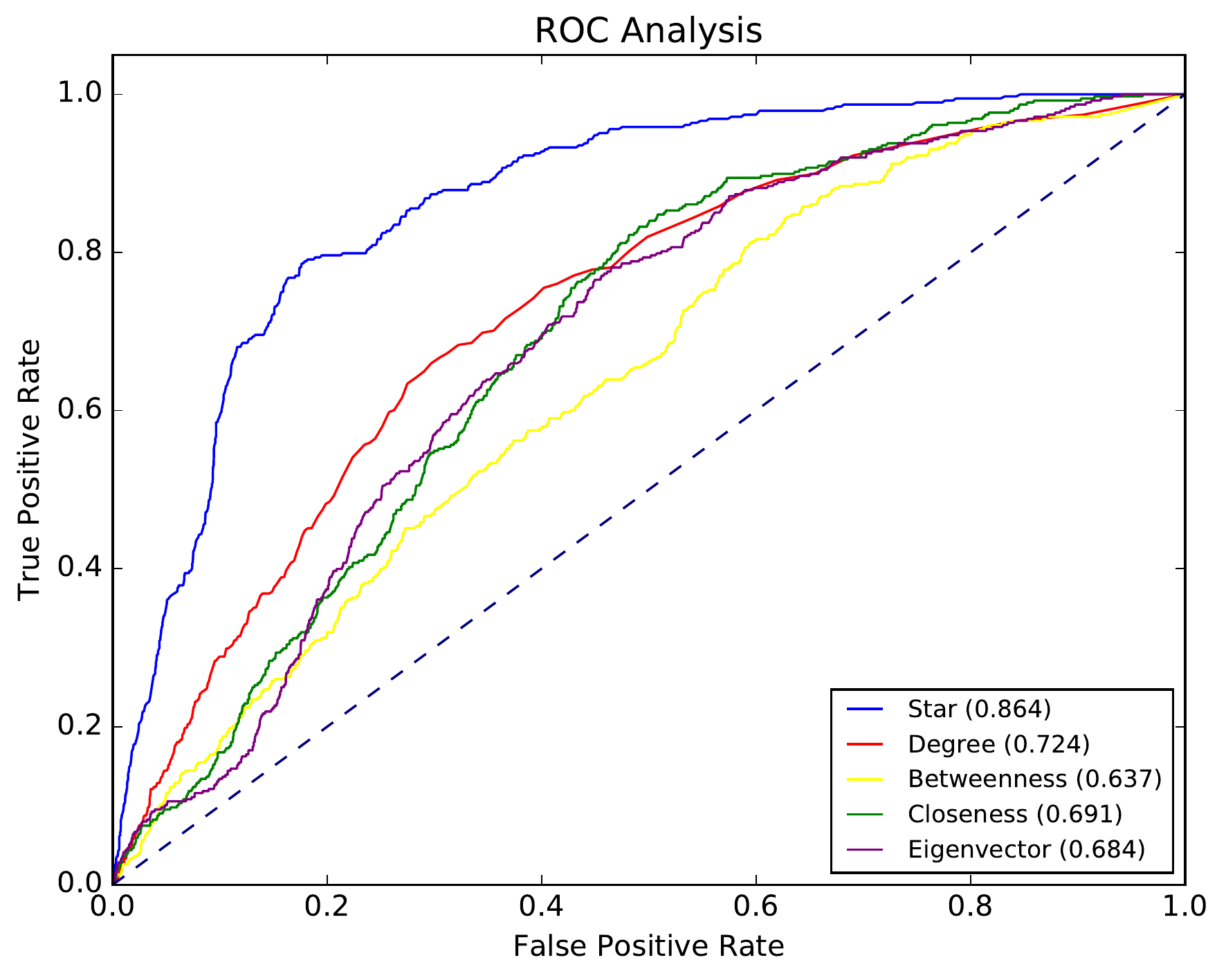}  
	}
	\caption{The Receiver Operating Characteristic curves for each metric for the \emph{Salmonella enterica CT18} (yeast) and the {\em Caenorhabditis elegans} organisms with a threshold of 70\%. \label{ROC7002}}
\end{figure}

\begin{figure}
	\tiny
	\centering
	{\includegraphics[width=0.33\textwidth]{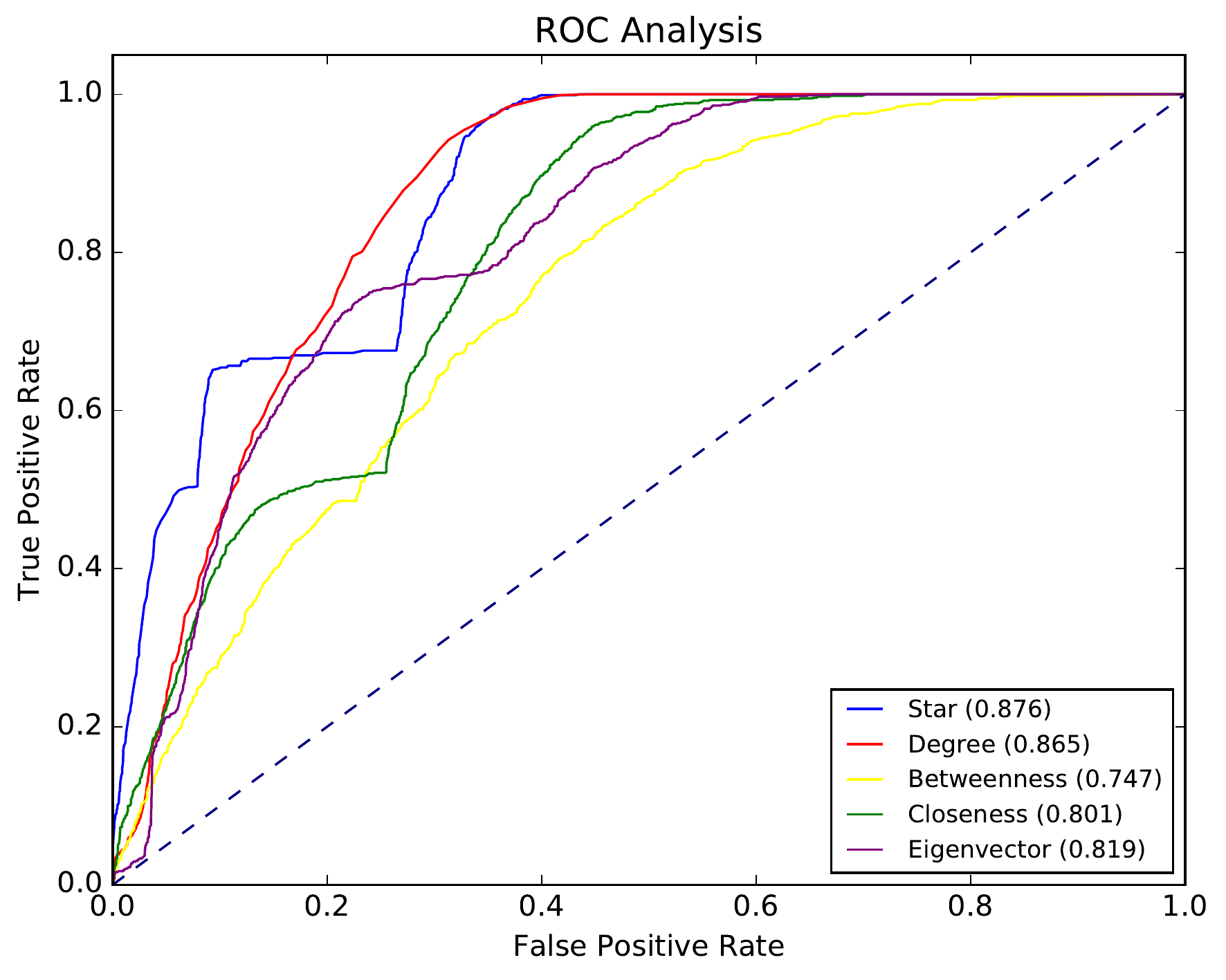}  \includegraphics[width=0.33\textwidth]{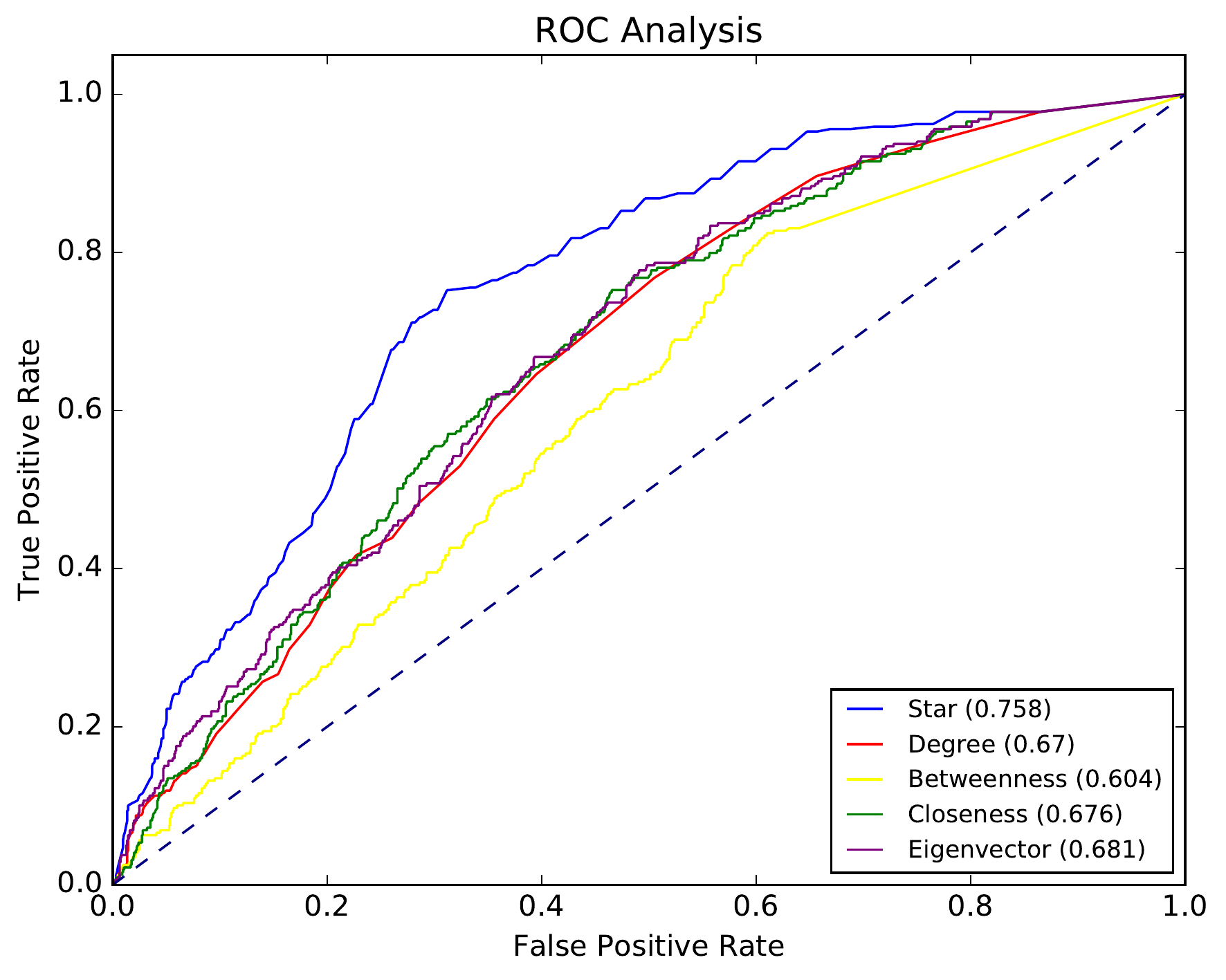}  \includegraphics[width=0.33\textwidth]{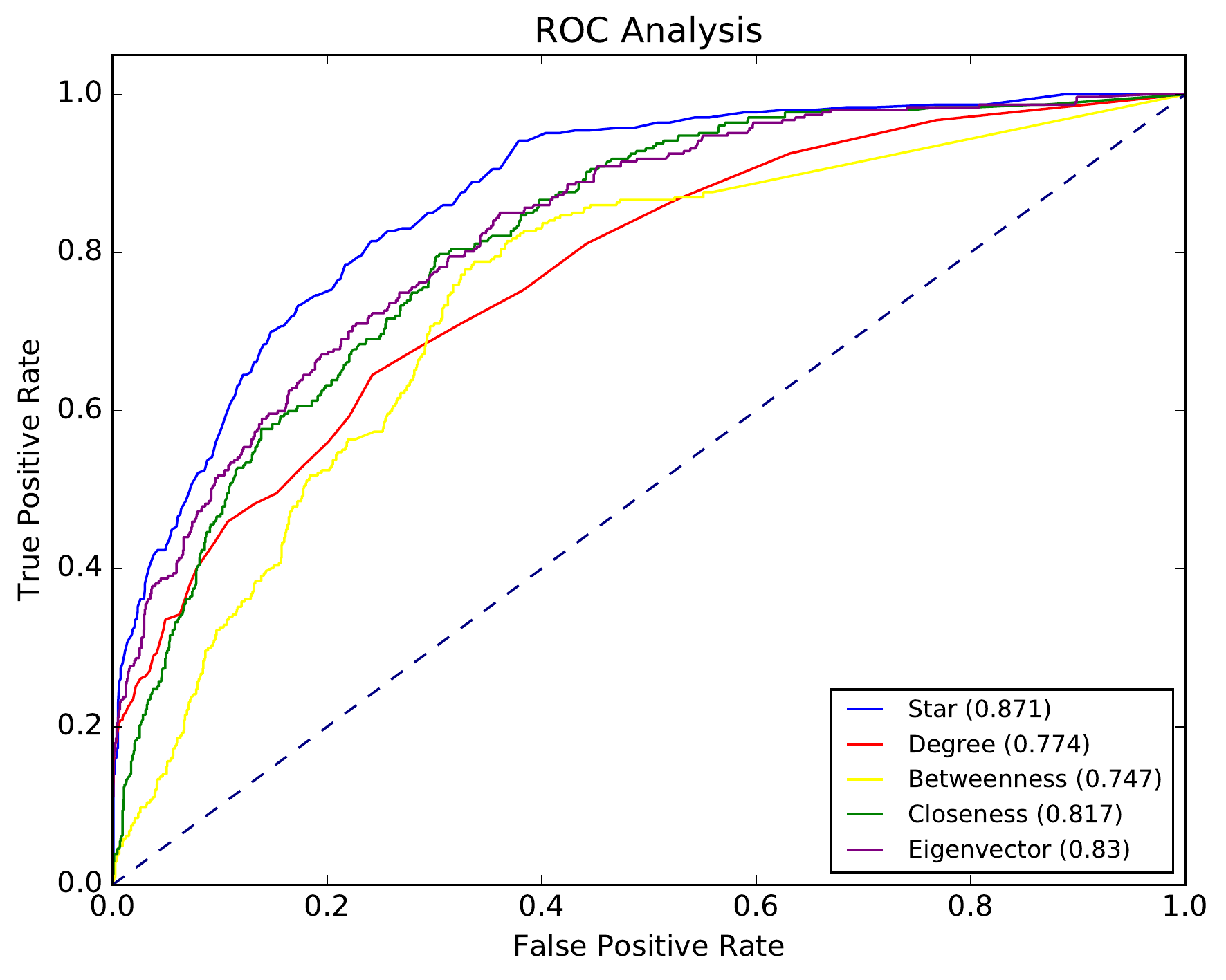}}
	\caption{The Receiver Operating Characteristic curves for each metric for the \emph{Saccharomyces cerevisiae} (yeast), the {\em Helicobacter pylori}, and the {\em Staphylococcus aureus} organisms with a threshold of 80\%. \label{ROC8001}}
\end{figure}

\begin{figure}
	\tiny
	\centering
	{\includegraphics[width=0.33\textwidth]{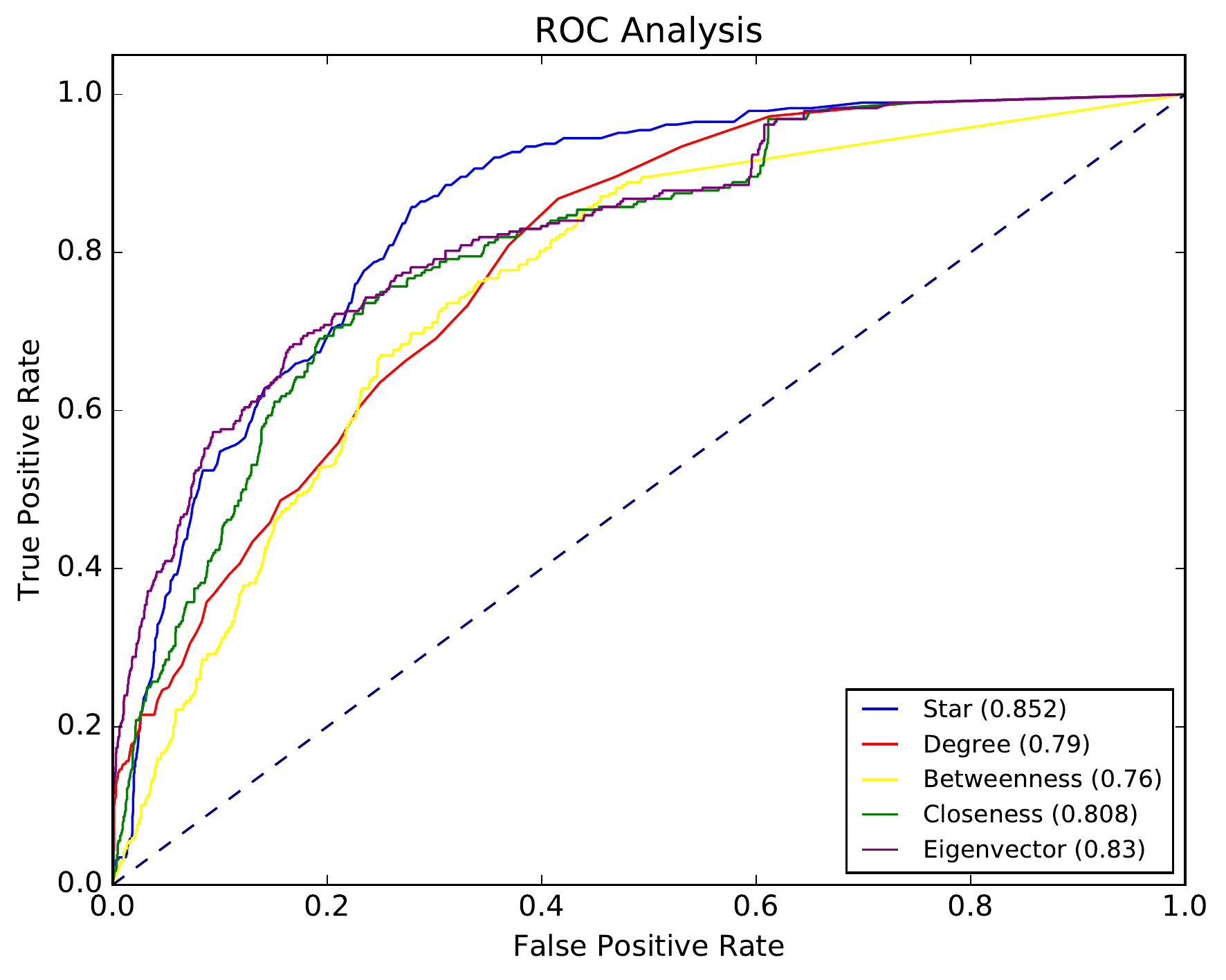}  \includegraphics[width=0.33\textwidth]{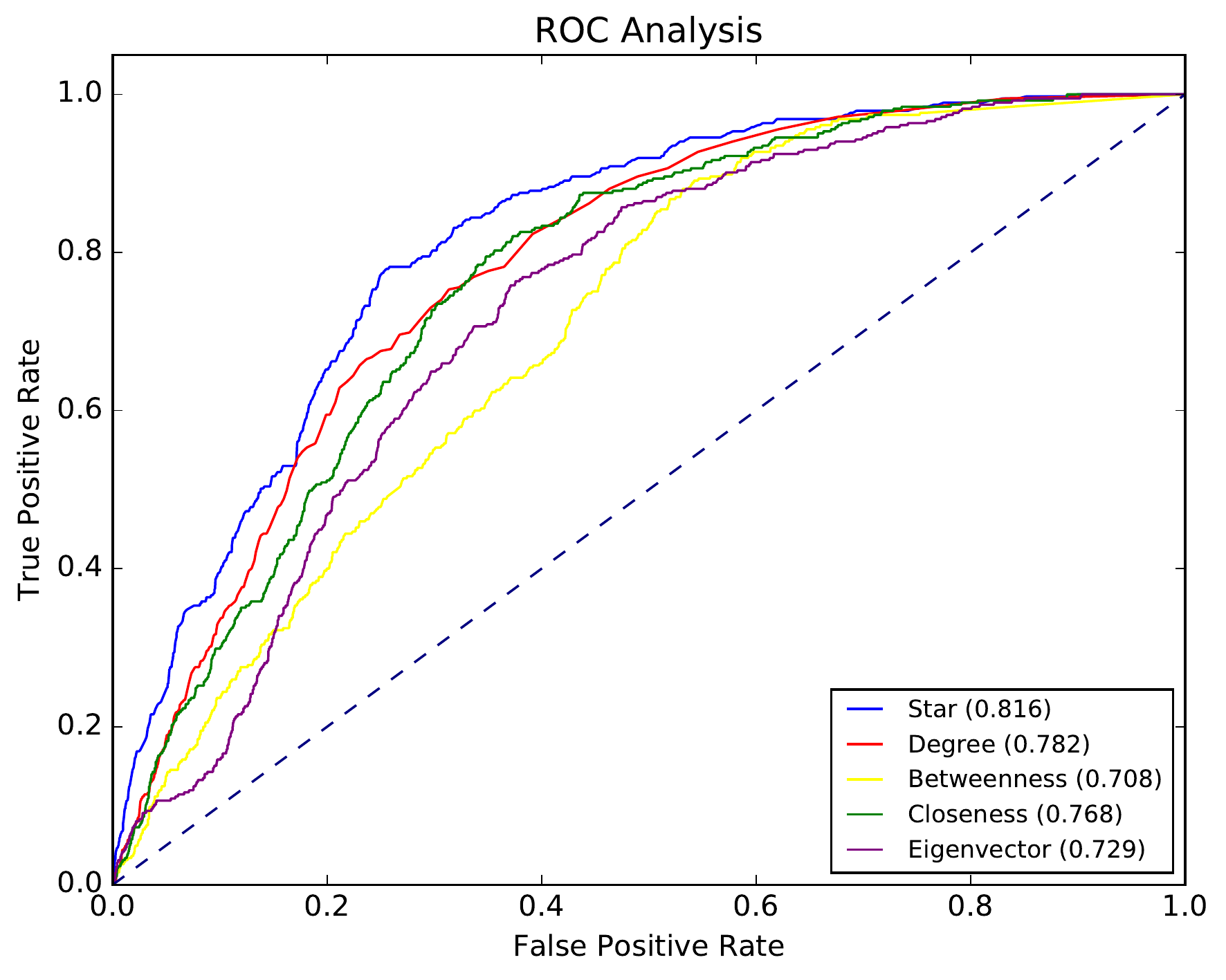}  
	}
	\caption{The Receiver Operating Characteristic curves for each metric for the \emph{Salmonella enterica CT18} (yeast) and the {\em Caenorhabditis elegans} organisms with a threshold of 80\%. \label{ROC8002}}
\end{figure}

 
%
%


\bibliography{bibliography} 


\end{document}